\documentclass{article}

\usepackage{amsmath}
\usepackage{graphicx}
\usepackage{amssymb}
\usepackage{amsthm}
\usepackage{caption}
\usepackage{subcaption}
\usepackage{bbm}
\usepackage{mathtools}
\usepackage{algorithmic}
\usepackage{courier} 
\usepackage{algorithm}
\graphicspath{{./figures/}}
\usepackage[round]{natbib}

\usepackage{chngcntr}
\usepackage{apptools}
\AtAppendix{\counterwithin{thm}{section}}

\usepackage{hyperref}
\usepackage{xcolor}

\numberwithin{equation}{section}
\newcommand{\ee}{{\rm e}\hspace{1pt}}

\newcommand{\dd}{\hspace{1pt}{\rm d}\hspace{0.5pt}}

\newcommand{\veps}{\varepsilon}
\newtheorem{thm}{Theorem}
\newtheorem{lem}[thm]{Lemma}
\newtheorem{cor}[thm]{Corollary}
\newtheorem{defn}[thm]{Definition}

\newcommand{\B}{ b }

\newcommand{\viewS}[1]{\text{View}_{\mathcal M}^{A_s}(#1)}
\newcommand{\view}[1]{\text{View}_{\mathcal M}^{A_w}(#1)}
\newcommand{\deltas}[1]{\delta_{#1}(s)}
\newcommand{\deltat}[1]{\delta_{#1}(t)}
\title{Tight Accounting in the Shuffle Model of  Differential Privacy}

\author{Antti Koskela$^{1,2}$, Mikko Heikkil\"a$^{3}$ and Antti Honkela$^{2}$ \vspace{5mm} \\
$^1$ Nokia Bell Labs, Espoo, Finland \\
$^2$ Helsinki Institute for Information Technology HIIT,\\
    Department of Computer Science, University of Helsinki, Finland \\
  $^3$ Helsinki Institute for Information Technology HIIT,\\
     Department of Mathematics and Statistics, University of Helsinki, Finland }

\date{}

\usepackage{graphicx}

\begin{document}

\maketitle

\begin{abstract}

Shuffle model of differential privacy is a novel distributed privacy model based on a combination of local privacy mechanisms and a secure shuffler. 
It has been shown that the additional randomisation provided by the shuffler improves privacy bounds compared to the purely local mechanisms. 
Accounting tight bounds, however, is complicated by the complexity brought by the shuffler.
The recently proposed numerical techniques for evaluating $(\varepsilon,\delta)$-differential privacy guarantees 
have been shown to give tighter bounds than commonly used methods for compositions of various complex mechanisms. 
In this paper, we show how to obtain accurate bounds for adaptive compositions of general $\veps$-LDP shufflers using the analysis by Feldman et al. (2021) and tight bounds for adaptive compositions of shufflers of $k$-randomised response mechanisms,
using the analysis by Balle et al. (2019). We show how to speed
up the evaluation of the resulting privacy loss distribution from $\mathcal{O}(n^2)$
to $\mathcal{O}(n)$, where $n$ is the number of users, without noticeable change in the resulting $\delta(\veps)$-upper bounds. 
We also demonstrate looseness of the existing bounds and methods found in the literature, improving previous composition results significantly. 

\end{abstract}

\section{Introduction}

The shuffle model of differential privacy (DP) is a distributed privacy model which sits between the high trust-high utility centralised DP, and the low trust-low utility local DP (LDP). In the shuffle model, the individual results from local randomisers are only released through a secure shuffler. This additional randomisation leads to ``amplification by shuffling'', resulting in better privacy bounds against adversaries without access to the unshuffled local results.

We consider computing privacy bounds for both single and composite shuffle protocols, where by composite protocol we mean a protocol, where the subsequent user-wise local randomisers depend on the same local datasets and possibly on the previous output of the shuffler, and at each round the results from the local randomisers are independently shuffled. Moreover, using the analysis by~\citet{feldman2021hiding},
we provide bounds in the case the subsequent local randomisers are allowed to depend adaptively
on the output of the previous ones.

In this paper we show how numerical accounting \citep{koskela2020,koskela2021tight,gopi2021} can be employed for 
tight privacy analysis of both single and composite shuffle DP mechanisms. 
To our knowledge, ours is the only existing method enabling tight 
privacy accounting for composite protocols in the shuffle model. 
We demonstrate that thus obtained bounds are always tighter than the existing bounds from the literature.

By using the tight privacy bounds we can also evaluate how significantly  adversaries with varying capabilities differ in terms of the resulting privacy bounds. That is, we can %effectively
quantify the value of information in terms of privacy by comparing tight privacy bounds under varying assumptions.

\subsection{Related work}

DP was originally defined in the central model assuming a trusted aggregator by \citet{dwork_et_al_2006}, while the fully distributed LDP was formally introduced and analysed by \citet{kasiviswanathan2011}.
Closely related to the shuffle model of DP, \citet{Bittau2017} proposed the Encode, Shuffle, Analyze framework for distributed learning, which uses the idea of secure shuffler for enhancing privacy. The shuffle model of DP was formally defined by \citet{cheu2019distributed}, who also provided the first separation result showing that the shuffle model is strictly 
between the central and the local models of DP. 
Another direction initiated by \citet{cheu2019distributed} and continued, e.g., by \citet{balle2020multi, ghazi2021} has established a separation between single- and multi-message shuffle protocols. %Characterising the  exact nature of these separations has been the aim of many 
%subsequent works as well, as demonstrated by a recent survey \citep{cheu2020}.

There exists several papers on privacy amplification by shuffling, some of which are central to this paper. \citet{Erlingsson2019} showed that the introduction of a secure shuffler amplifies the privacy guarantees against an adversary, who is not able to access the outputs from the local randomisers but only sees the shuffled output. 
\citet{balle2019blanket} 
improved the amplification results and introduced the idea of privacy blanket, 
which we also utilise in our analysis of $k$-randomised response in Section~\ref{sec:kRR}. We compare our bounds with those of \citet{balle2019blanket} in Section~\ref{sec:krr_FA}. 
\citet{feldman2021hiding} used a related idea of hiding in the crowd to improve on the previous results, while \citet{girgis2021shuffled} generalised shuffling amplification further to scenarios with composite protocols 
and parties with more than one local sample under simultaneous communication and privacy restrictions. We use some results of \citet{feldman2021hiding} 
in the analysis of general LDP mechanisms,
and compare our bounds with theirs in Section~\ref{sec:experiments_feldman}. 
%\textbf{could be in discussion:}
We also calculate 
privacy bounds in the setting considered by~\citet{girgis2021shuffled}, namely
in the case a fixed subset of users sending contributions to the shufflers are sampled randomly.
This can be seen as a %so called 
subsampled mechanism and we are able to combine the analysis
of~\citet{feldman2021hiding}, the PLD related subsampling results of~\citet{zhu2021optimal}
and FFT accounting to %experimentally 
obtain tighter $(\varepsilon,\delta)$-bounds than~\citet{girgis2021shuffled}, as shown in Section~\ref{sec:experiments_girgis}.

\section{Background}
\label{sec:background}

Before analysing the shuffled mechanisms we need to 
introduce some theory and notations. With apologies for conciseness, 
we start by 
defining DP and PLD, and finish with the Fourier accountant.
For more details, we refer to~\citep{koskela2021tight,gopi2021,zhu2021optimal}.

\subsection{Differential privacy and privacy loss distribution}

An input data set containing $n$ data points is denoted as $X = (x_1,\ldots,x_n)
\in \mathcal{X}^n$, where $x_i \in \mathcal{X}$, $1 \leq i \leq n $.
We say $X$ and $X'$ are neighbours if we get one by substituting
	one element in the other (denoted $X \sim X'$).

\begin{defn} \label{def:indistinguishability}
	Let $\varepsilon > 0$ and $\delta \in [0,1]$.
	Let $P$ and $Q$ be two random variables taking values in the
	same measurable space $\mathcal{O}$.
	We say that $P$ and $Q$ are
	$(\varepsilon,\delta)$-indistinguishable, denoted $P \simeq_{(\veps,\delta)} Q$,
	if for every measurable set $E \subset \mathcal{O}$ we have
	\begin{equation*}
		\begin{aligned}
			&\mathrm{Pr}( P \in E ) \leq \ee^\varepsilon \mathrm{Pr} (Q \in E ) + \delta, \\
			&\mathrm{Pr}( Q \in E ) \leq \ee^\varepsilon \mathrm{Pr} (P \in E ) + \delta.
		\end{aligned}
	\end{equation*}
\end{defn}

\begin{defn} \label{def:dp}
	Let $\varepsilon > 0$ and $\delta \in [0,1]$. %Let $\sim$ define the neighbouring relation.
	Mechanism $\mathcal{M} \, : \, \mathcal{X}^n \rightarrow \mathcal{O}$ is  $(\veps, \delta)$-DP
	%$(\varepsilon,\delta,\sim)$-DP % or $(\varepsilon,\delta,\sim_S)$-DP  
	if for every $X \sim X'$:
	%$X \sim X'$:
	%and every measurable $E \subset \mathcal{R}$:
	$ \mathcal{M}(X) \simeq_{(\veps,\delta)} \mathcal{M}(X')$.
	We call $\mathcal{M}$ tightly 	$(\veps,\delta)$-DP, 
% 	$(\veps,\delta,\sim)$-DP, 
	if there does not exist $\delta' < \delta$
	such that $\mathcal{M}$ is $(\veps,\delta')$-DP.
	The case when $n=1$ and $\delta=0$ is called $\varepsilon$-LDP.
\end{defn}

Tight DP bounds can also be characterised as
\small
\begin{equation*} %\label{eq:hockey}
\delta(\veps) = \max_{X \sim X'}\{ H_{\ee^\veps}(\mathcal{M}(X)||\mathcal{M}(X')), H_{\ee^\veps}(\mathcal{M}(X')||\mathcal{M}(X)) \},	
\end{equation*}
\normalsize 
where for $\alpha>0$ the Hockey-stick divergence is defined as
\begin{equation*}% \label{eq:alphadiv}
	H_\alpha(P||Q) = \int \max \{0, P(t) - \alpha \cdot Q(t) \} \, \dd t.
\end{equation*}

%In this work, %finding the pair of outputs $\mathcal{M}(X)$ and $\mathcal{M}(X')$ that give the maximum $\delta(\varepsilon)$
%will be clear from the context, and thus 
We can generally find tight $(\veps,\delta)$-bounds by analysing
a tightly dominating pair of random variables or distributions:
%corresponding to neighbouring data sets.
\begin{defn}[\citealt{zhu2021optimal}]
A pair of distributions $(P,Q)$ 
is a \emph{dominating pair} of distributions for mechanism $\mathcal{M}(X)$ if
for all neighbouring datasets $X$ and $X'$ and for all $\alpha>0$,
$$
H_\alpha(\mathcal{M}(X) || \mathcal{M}(X')) \leq H_\alpha(P || Q).
$$
If the equality holds for all $\alpha$ for some $X, X'$, then $(P,Q)$ is  tightly dominating.
\end{defn}

%\note{We actually also consider 2D distributions?}

We analyse discrete-valued distributions, which means 
that a dominating pair of distribution $(P,Q)$ %will be of 
can be described by a generalised  probability density functions as
\begin{equation} \label{eq:delta_sum}
	\begin{aligned}
		&P(t) = \sum\nolimits_i a_{P,i} \cdot \deltat{t_{P,i}},  \\
		& Q(t) = \sum\nolimits_i a_{Q,i} \cdot \deltat{t_{Q,i}},
	\end{aligned}
\end{equation}
where $\delta_t( \cdot )$, $t \in \mathbb{R}^d$, 
denotes the Dirac delta function centred at $t$, and $t_{P,i},t_{Q,i} \in \mathbb{R}^d$ and $a_{P,i},a_{Q,i} \geq 0$.
The PLD determined by a pair $(P,Q)$ is defined as follows.
\begin{defn} \label{def:pld}
Let $P$ and $Q$ be generalised probability density functions as defined by \eqref{eq:delta_sum}.
%$\mathcal{M}(X)$ and $\mathcal{M}(Y)$, respectively, both being of the form \eqref{eq:delta_sum}.
We define the generalised privacy loss distribution (PLD) $\omega_{P/Q}$ as   % for $S \subset \mathbb{R}$ as
\begin{equation*}% \label{eq:omega_pld}
	\begin{aligned}
	\omega_{P/Q}(s) &= \sum\nolimits_{{t_{P,i} = t_{Q,j} }}   a_{P,i} \cdot \deltas{ s_{i,j}}, 
	\quad s_{i,j} = \log \left( \frac{a_{P,i}}{a_{Q,j}} \right).
	\end{aligned}
\end{equation*}
%s_i = \log \left( \tfrac{a_{X,i}}{a_{X',j}} \right)
%where $$.
\end{defn}

The following theorem~\citep[Thm.\;10]{zhu2021optimal}
% \citep[see e.g.][]{sommer2019privacy,koskela2021tight}
shows that the tight $(\veps,\delta)$-bounds for compositions 
of adaptive mechanisms are obtained using convolutions of PLDs. 
The expression \eqref{eq:pld_integral} is equivalent to the hockey-stick divergence 
$H_\veps(P || Q)$~\citep[see e.g.][]{sommer2019privacy,koskela2021tight,gopi2021}. 
% Further on, based on the results of \citet{zhu2021optimal} we then show that we can use Thm.~\ref{thm:integral} also to calculate tight bounds for the adaptive mechanisms we consider.
\begin{thm} \label{thm:integral}
Consider an $n_c$-fold adaptive composition given a (tightly) dominating pair $(P,Q)$.
%of a mechanism $\mathcal{M}$. 
The composition is (tightly) $(\veps,\delta)$-DP for $\delta(\veps)$ given by
\begin{equation}  \label{eq:pld_integral}
	\begin{aligned}
		\delta_{P/Q}(\veps)  = 1 - \big(1-\delta_{P/Q}(\infty)\big)^{n_c} + 
		\int\nolimits_\veps^\infty (1 - \ee^{\veps - s})\left(\omega_{P/Q} *^{n_c} \omega_{P/Q} \right) (s)  \, \dd s,
	\end{aligned}
\end{equation}
$$
\delta_{P/Q}(\infty) =
 \sum\nolimits_{ \{ t_i \, :  \, \mathbb{P}( Q = t_i) = 0 \} } 
\mathbb{P}( P = t_i)
$$
and  $\omega_{P/Q} *^{n_c} \omega_{P/Q}$ denotes the $n_c$-fold convolution of 
the generalised density function $\omega_{P/Q}$.
%(an analogous expression holds for $\delta_{Q/P}(\veps)$).
\end{thm}
%We shortly mention that 
When computing tight $\delta(\veps)$-bounds for the shufflers of the $k$-RR local randomisers, instead of \eqref{eq:pld_integral}, for a certain distribution $\omega$ determined
by the shuffler mechanism, we need to evaluate expressions of the form 
\begin{equation}  \label{eq:pld_integral2}
	\begin{aligned}
		\delta(\veps)  =  1 - \big(1-\delta(\infty)\big)^{n_c} 
		 + 
		\int\limits_\veps^\infty \left(\omega *^{n_c} \omega\right) (s)  \, \dd s,
	\end{aligned}
\end{equation}
where $\delta(\infty) = 1 - \sum_i \omega(i)$. The FFT-based numerical accounting is straightforwardly applied to \eqref{eq:pld_integral2}
as well.

%%%%%%%%%%%%%%%%%%%%%%%%%%%%%%%%%%%%%%%%%%%%%%%%%%%%%%%%%%%%%%%%%%%%%%%%%%%%%%%%%%%%%%%%%%%%%%%%%%%%%%%%%%%%%%%%%%%%%%%%
\subsection{Numerical Evaluation of DP Parameters Using FFT} %of Discrete Mechanisms}
\label{sec:FA}
%%%%%%%%%%%%%%%%%%%%%%%%%%%%%%%%%%%%%%%%%%%%%%%%%%%%%%%%%%%%%%%%%%%%%%%%%%%%%%%%%%%%%%%%%%%%%%%%%%%%%%%%%%%%%%%%%%%%%%%%

In order to evaluate integrals of the form \eqref{eq:pld_integral} and \eqref{eq:pld_integral2} and to find tight privacy bounds, 
we use the Fast Fourier Transform (FFT)-based method by~\citet{koskela2020,koskela2021tight}
called the Fourier Accountant (FA).
%(see \citealt[Ch.5.4]{koskela2021heterogeneous} for computational complexity bounds).
This means that we truncate and place the PLD $\omega$ on an equidistant numerical grid over an
interval $[-L,L]$, $L>0$.
Convolutions are evaluated using the FFT algorithm and using the error analysis
the error incurred by the method can be bounded. We note that alternatively, for accurately computing the integrals and obtaining tight $\delta(\veps)$-bounds, we could also use the FFT-based method proposed by~\citet{gopi2021}.

In the next sections we construct the PLD $\omega$ for different shuffling mechanisms.
In practice this means that in each case we need a dominating pair of random variables $P$ and $Q$ 
that then lead to an $(\veps,\delta)$-DP bound.

%%%%%%%%%%%%%%%%%%%%%%%%%%%%%%%%%%%%%%%%%%%%%%%%%%%%%%%%%%%%%%%%%%%%%%%%%%%%%%%%%%%%%%%%%%%%%%%%
%%%%%%%%%%%%%%%%%%%%%%%%%%%%%%%%%%%%%%%%%%%%%%%%%%%%%%%%%%%%%%%%%%%%%%%%%%%%%%%%%%%%%%%%%%%%%%%%
%%%%%%%%%%%%%%%%%%%%%%%%%%%%%%%%%%%%%%%%%%%%%%%%%%%%%%%%%%%%%%%%%%%%%%%%%%%%%%%%%%%%%%%%%%%%%%%%
%\section{General analysis via clones of $\veps_0$-LDP local randomisers}
\section{General shuffled $\veps_0$-LDP mechanisms}\label{pld_for_shuffling}
%%%%%%%%%%%%%%%%%%%%%%%%%%%%%%%%%%%%%%%%%%%%%%%%%%%%%%%%%%%%%%%%%%%%%%%%%%%%%%%%%%%%%%%%%%%%%%%%
%%%%%%%%%%%%%%%%%%%%%%%%%%%%%%%%%%%%%%%%%%%%%%%%%%%%%%%%%%%%%%%%%%%%%%%%%%%%%%%%%%%%%%%%%%%%%%%%
%%%%%%%%%%%%%%%%%%%%%%%%%%%%%%%%%%%%%%%%%%%%%%%%%%%%%%%%%%%%%%%%%%%%%%%%%%%%%%%%%%%%%%%%%%%%%%%%

\citet{feldman2021hiding} consider general $\veps_0$-LDP local randomisers combined with a shuffler.
The analysis allows also sequential adaptive compositions of the user contributions before shuffling. % e.g. a novel analysis of the DP stochastic gradient descent.
The analysis is based on decomposing individual LDP contributions to mixtures of data dependent part and noise, which leads to finding $(\veps,\delta)$-bound for the 
2-dimensional distributions~\citep[see Thm.\;3.2 of][]{feldman2021hiding}
\begin{equation} \label{eq:2n}
	\begin{aligned}
		P &= (A + \Delta, C-A+1-\Delta), \\ 
		Q &= (A+1-\Delta,C-A+\Delta),
	\end{aligned}
\end{equation}
where for $n \in \mathbb{N}$,
\begin{equation*}% \label{eq:CADelta}
	\begin{aligned}
		C \sim \mathrm{Bin}(n-1,\ee^{-\veps_0}), \quad A \sim \mathrm{Bin}(C,\tfrac{1}2), \quad
	 \Delta \sim \mathrm{Bern}\left(\tfrac{\ee^{\veps_0}}{\ee^{\veps_0}+1}\right).
	\end{aligned}
\end{equation*} 
Intuitively, $C$ denotes the number of other users whose mechanism outputs are indistinguishable ``clones'' of the two different users with $A$ denoting random split between these.
Moreover, a numerical method to compute the hockey-stick divergence $H_{\ee^\veps}(P || Q)$ is proposed.
Using the results of~\citet{zhu2021optimal} and the following observation, we can use the Fourier accountant to obtain accurate bounds also
for adaptive compositions of general $\veps_0$-LDP shuffling mechanisms:
\begin{lem} \label{lem:PQcomp}
Let $X$ and $X'$ be neighbouring datasets and denote by $\mathcal{A}_s(X)$ and $\mathcal{A}_s(X')$ outputs of the shufflers of 
adaptive $\veps_0$-LDP local randomisers~\citep[for more detailed description, see Thm.\;3.2 of][]{feldman2021hiding}.
Then, for all $\veps>0$,
$$
H_{\ee^\veps}(\mathcal{A}_s(X) || \mathcal{A}_s(X')) \leq H_{\ee^\veps}(P || Q),
$$
where $P$ and $Q$ are given as in \eqref{eq:2n}.
\begin{proof}
By Thm.\;3.2 of~\citet{feldman2021hiding} there exists a post-processing algorithm $\Phi$ such that 
$\Phi(\mathcal{A}_s(X))$ is distributed identically to $P$ and $\Phi(\mathcal{A}_s(X'))$ identically to $Q$.
Since in the construction of Thm.\;3.2 of~\citet{feldman2021hiding} $X$ and $X'$ can be any neighbouring datasets, the claim follows from 
the post-processing property of DP~\citep[see Proposition 2.1 in][]{DworkRoth}.
\end{proof}
\end{lem}

Using Lemma 46 of \citet{zhu2021optimal} and the above Lemma~\ref{lem:PQcomp} yields the following result:
\begin{cor}
The pair of distributions $(P,Q)$ in \eqref{eq:2n}
is a dominating pair of distributions for the shuffling mechanism $\mathcal{A}_s(X)$.
\end{cor}

Furthermore, using Thm. 10 of~\citet{zhu2021optimal}, we can bound the $\delta(\veps)$ of $n_c$-wise adaptive composition of
the shuffler $\mathcal{A}_s$ using product distributions of $P$s and $Q$s:
\begin{cor} \label{cor:pld_compositions}
Denote $\mathcal{A}^{n_c}_s(X,z_0) = \mathcal{A}_s(X,\mathcal{A}_s(X,...\mathcal{A}_s(X,z_0)))$ for some initial state $z_0$. For all neighbouring datasets $X$ and $X'$
and for all $\alpha>0$,
\begin{equation} \label{eq:DalphaPQ}
	H_\alpha(\mathcal{A}^{n_c}_s(X) || \mathcal{A}^{n_c}_s(X')) 
	\leq H_\alpha(P \times \ldots \times P || Q \times \ldots \times Q),
\end{equation}
where $P \times \ldots \times P$ and $Q \times \ldots \times Q$ are $n_c$-wise product distributions. 
\end{cor}

The case of heterogeneous adaptive compositions (e.g. for varying $n$ and $\veps_0$) can be handled analogously using 
Thm. 10 of~\citet{zhu2021optimal}.

%~\citet{feldman2021hiding} also give a numerical method for obtaining a tight $(\veps,\delta)$-bound. 
Thus, using \eqref{eq:DalphaPQ} for $\alpha=\ee^\veps$, we get upper bounds for adaptive compositions of general shuffled $\veps_0$-LDP  mechanisms 
with the Fourier accountant by finding the PLD for the distributions $P,Q$ (given in Eq.~\eqref{eq:2n}). 
Note that even though the resultsing $(\veps,\delta)$-bound is tight for $P$'s and $Q$'s, 
it need not be tight for a specific mechanism like the shuffled $k$-RR.
The bound simply gives an upper bound for any shuffled $\veps_0$-LDP mechanisms.
In the Supplements we give also comparisons of the tight bounds obtained with $P$ and $Q$ of \eqref{eq:2n}
and with those of the strong $k$-RR adversary (Sec.~\ref{sec:kRR}).

%\subsection{Implementing the approach by~\citet{feldman2021hiding} using the Fourier accountant} 
\subsection{PLD for shuffled $\veps_0$-LDP mechanisms}
\label{subsec:pld_clones}

As already noted, we can find $\delta(\veps)$-upper bounds for general shuffled $\veps_0$-LDP
mechanisms by analysing the pair of distributions $(P,Q)$ of Eq.~\eqref{eq:2n}. 
To analyse the compositions, we need to determine the PLD $\omega_{P/Q}$. 
Since this is straight-forward but the details are messy, 
we simply state the result here and give the details in the Supplement.

Denoting $q=\frac{\ee^{\veps_0}}{\ee^{\veps_0}+1}$, we see that the distributions in \eqref{eq:2n} are given by the mixture distributions
\begin{equation*}% \label{eq:wtPQ}
\begin{aligned}
	P &= q \cdot P_1 + (1-q) \cdot P_0, \\
	 Q &= (1-q) \cdot P_1 + q \cdot P_0,
\end{aligned}
\end{equation*}
where
\begin{equation*} \label{eq:P1P0Q1}
P_1 ~ (A+1,C-A), \quad P_0 ~ (A,C-A+1).
\end{equation*}

% Using these expressions, and the fact that $\mathbb{P}(P_0=(a,0))=0$ for all $a$ and
% $\mathbb{P}(P_1=(0,b))=0$ for all $b$, we get the following expressions needed for $s_{a,b}$'s.
In the Supplements we show the following expressions that will determine the PLD.
\begin{lem} %\label{Alem:logPQ}
When $b>0$ and $a \geq 0$,
\begin{equation*} % \label{Aeq:logwtPwtQ}
	  \frac{ \mathbb{P}(P=(a,b)) }{ \mathbb{P}(Q=(a,b)) } 
	= \frac{ q \cdot \frac{a}{b} + (1-q) }{q + (1-q)\frac{a}{b} }.
\end{equation*}
When $0 < a \leq n$,
$$
\frac{ \mathbb{P}(P=(a,0)) }{ \mathbb{P}(Q=(a,0)) } = \frac{q}{1-q}.
$$
% When $0 < b \leq n$,
% $$
% \frac{ \mathbb{P}(P=(0,b)) }{ \mathbb{P}(Q=(0,b)) } = \frac{1-q}{q}.
% $$
\end{lem}
\begin{lem} \label{Alem:P1P0_2}
When $a>0$,
\begin{equation*} % \label{eq:pab}
    \begin{aligned}
    \mathbb{P}(P_1=(a,b)) = {n-1 \choose i}  {i \choose j} \ee^{-i \cdot \veps_0} (1-\ee^{-\veps_0})^{n-1-i} \frac{1}{2^i},
    \end{aligned}
\end{equation*}
where $(a,b) = (j+1,i-j)$ (i.e., $C=i$ and $A=j$), and
\begin{equation*} % \label{eq:pab2}
 \mathbb{P}(P_0=(a,b))  = \frac{\ee^{-\veps_0}}{1 - \ee^{-\veps_0}} \frac{n-a-b}{2 a}  \mathbb{P}(P_1=(a,b)).
\end{equation*}
For $0 < b \leq n$, $\mathbb{P}(P_1=(0,b))=0$ and
$$
\mathbb{P}(P_0=(0,b)) = {n-1 \choose b-1} \left( \frac{\ee^{-\veps_0}}{2}\right)^{b-1} (1-\ee^{-\veps_0})^{n-b}.
$$
% \begin{proof}
% The expressions follow directly from the definitions of $P_0$, $P_1$, $A$ and $C$.
% \end{proof}
\end{lem}
These expressions together give the PLD 
\begin{equation} \label{eq:clones_pld}
	\begin{aligned}
		 \omega_{ P/Q }(s) = \sum\nolimits_{a,b} \mathbb{P}(P=(a,b)) \cdot \delta_{s_{a,b}}(s), \quad
		 s_{a,b} = \log \left( \frac{ \mathbb{P}(P=(a,b)) }{ \mathbb{P}(Q=(a,b)) }   \right),
	\end{aligned}
\end{equation}
and allow computing $\delta(\veps)$ using FFT.

\subsection{Lowering PLD computational complexity using Hoeffding's inequality}
\label{sec:clones_eff_approximation}

The PLD~\eqref{eq:clones_pld} has $\mathcal{O}(n^2)$ terms which makes its evaluation expensive for large number of users $n$.
Empirically, we find that the $\mathcal{O}(n^2)$-cost of forming the PLD dominates the cost of FFT already for $n=1000$.
Notice that the cost of FFT depends only on the number of grid points used for FFT, not on $n$.
Using an appropriate tail bound (Hoeffding) for the binomial distribution, we can neglect part of the mass and simply add it to $\delta_{P/Q}(\infty)$.
As $A$ is conditioned on $C$, we first use a tail bound on $C$ and then on $A$, to reduce the number of terms.
As a result we get an accurate approximation of $\omega_{P/Q}$ with only $\mathcal{O}(n)$ terms. We formalise this approximation as follows:

\begin{lem}
Let $\tau>0$ and denote $p=\ee^{- \veps_0}$. Consider the set
$$
S_n = \big[\max\big(0,(p-c_n)(n-1)\big), \min\big(n-1,(p+c_n)(n-1)\big)\big],
$$
where $c_n = \sqrt{\frac{\log (4/\tau)}{2(n-1)}}$ and the set
$$
\widehat{S}_i = \big[\max\big(0,( \tfrac{1}{2}-\widehat{c}_i) \cdot i \big), \min\big(n-1,(\tfrac{1}{2}+\widehat{c}_i) \cdot i \big)\big],
$$
where $\widehat{c}_i = \sqrt{\frac{\log (4/\tau)}{2 \cdot i}}$. Then, the distribution $\widetilde{\omega}_{P/Q}$ defined by
\begin{equation} \label{eq:wtomega}
	\begin{aligned}
		\widetilde{\omega}_{P/Q}(s) = \sum\nolimits_{i \in S_n} \sum\nolimits_{j \in \widehat{S}_i} \mathbb{P}\big(P=(j+1,i-j)\big) \cdot \delta_{s_{j+1,i-j}}(s), 
		\quad s_{a,b} = \log \left( \tfrac{ \mathbb{P}(P=(a,b)) }{ \mathbb{P}(Q=(a,b)) }   \right)
	\end{aligned}
\end{equation}
has $\mathcal{O}\big(n \cdot \log (4/\tau) \big)$ terms and differs from $\omega_{P/Q}$ at most mass $\tau$.
%, i.e.  $\int_\mathbb{R} \omega_{P/Q}(s) - \widetilde{\omega}_{P/Q}(s) \dd s \leq \tau$
\begin{proof}
Using Hoeffding's inequality for $C \sim \mathrm{Bin}(n-1,p)$ states that for $c>0$,
\begin{equation*}
	\begin{aligned}
		& \mathbb{P}\big( C \leq (p-c)(n-1) \big) \leq \exp\big( -2(n-1) c^2   \big),  \\
		& \mathbb{P}\big( C \geq (p+c)(n-1) \big) \leq \exp\big( -2(n-1) c^2   \big).
	\end{aligned}
\end{equation*}
Requiring that $2 \cdot \exp\left( -2(n-1) c^2\right) \leq \tau / 2$ gives the condition $c \geq \sqrt{\frac{\log (4/\tau)}{2(n-1)}}$
and the expressions for $c_n$ and $S_n$. Similarly, 
we use Hoeffding's inequality for $A \sim \mathrm{Bin}(C,\tfrac{1}{2})$ and
get expressions for $\widehat{c}_i$ and $\widehat{S}_i$. The total neglegted mass is at most 
$\tau/2 + \tau/2 = \tau$. For the number of terms, we see that $S_n$ contains at most 
$ 2 c_n (n-1) = \sqrt{n-1} \sqrt{2 \cdot \log (4/\tau)} $ terms and for each $i$, $\widehat{S}_i$ contains at most 
$ 2 \widehat{c}_i i = \sqrt{i} \sqrt{2 \cdot \log (4/\tau)} \leq \sqrt{n-1} \sqrt{2 \cdot \log (4/\tau)} $ terms. Thus $\widetilde{\omega}_{P/Q}$
has at most $\mathcal{O}(n \cdot \log (4/\tau) )$ terms.
We get the expression \eqref{eq:wtomega} by the change of variables $a =i+1$ ($A=i$) and $b=i-j$ ($C=j$).
\end{proof}
\end{lem}

When evaluating $\delta(\veps)$, we require that the neglected mass is smaller than some prescribed tolerance $\tau$ (e.g. $\tau=10^{-12}$),
and add it to $\delta_{P/Q}(\infty)$.
When computing guarantees for compositions, the cost of FFT, which only depends on the number of grid points, dominates the rest of the computation.

%...Denoting $p=\ee^{-\veps_0}$, and let $
%$$
%\mathbb{P}( \min (0,(p-c)(n-1)) \leq C \leq \max (n-1,(p+c)(n-1))   \leq \exp\big( -2n(p - c)^2   \big)
%$$
%to find a set $S$ with which we make an approximation $\sum_{i,j} \approx \sum_{(i,j) \in S}.$

\subsection{Experimental comparison to the numerical method of~\citet{feldman2021hiding}}
\label{sec:experiments_feldman}

%~\citet{feldman2021hiding} also give a numerical method for obtaining a tight $(\veps,\delta)$-bound. 

Figure~\ref{fig:comparison} shows a comparison between the PLD approach and the numerical method proposed by~\citet{feldman2021hiding}. 
%\note{maybe parameter values to supplement?}
%In the implementation of this method we use the parameter value $S=5$. 
We see that for a single composition the results given by this method
are not far from the results given by the Fourier Accountant (FA). 
This is expected as their method aims for 
giving an accurate upper bound for the %so-called
hockey-stick divergence between $P$ and $Q$, which is equivalent to what FA does. However, 
the method of~\citet{feldman2021hiding} only works for a single round, 
whereas FA also gives tight bounds for composite protocols. 
%In the Supplements we give results also for the cases $n=10^5,10^6$.
We emphasise here that FA gives strict upper $(\veps,\delta)$-bounds.
A downside of our approach is the slightly increased computational cost: for a single round protocol, evaluating tight bounds for $n=10^6$ took approximately 
4 times longer than using the method of~\citet{feldman2021hiding}, taking approximately one minute on a standard CPU. As the main cost of our approach consists of forming the PLD, 
the overhead cost of computing guarantees for compositions is small.
%We found that the numerical method by~\citet{feldman2021hiding} suffered from certain instabilities, difficult to get working for
%$\veps_0 = 6.0$, for example.

\begin{figure} [ht]
     \centering
        \includegraphics[width=.6\textwidth]{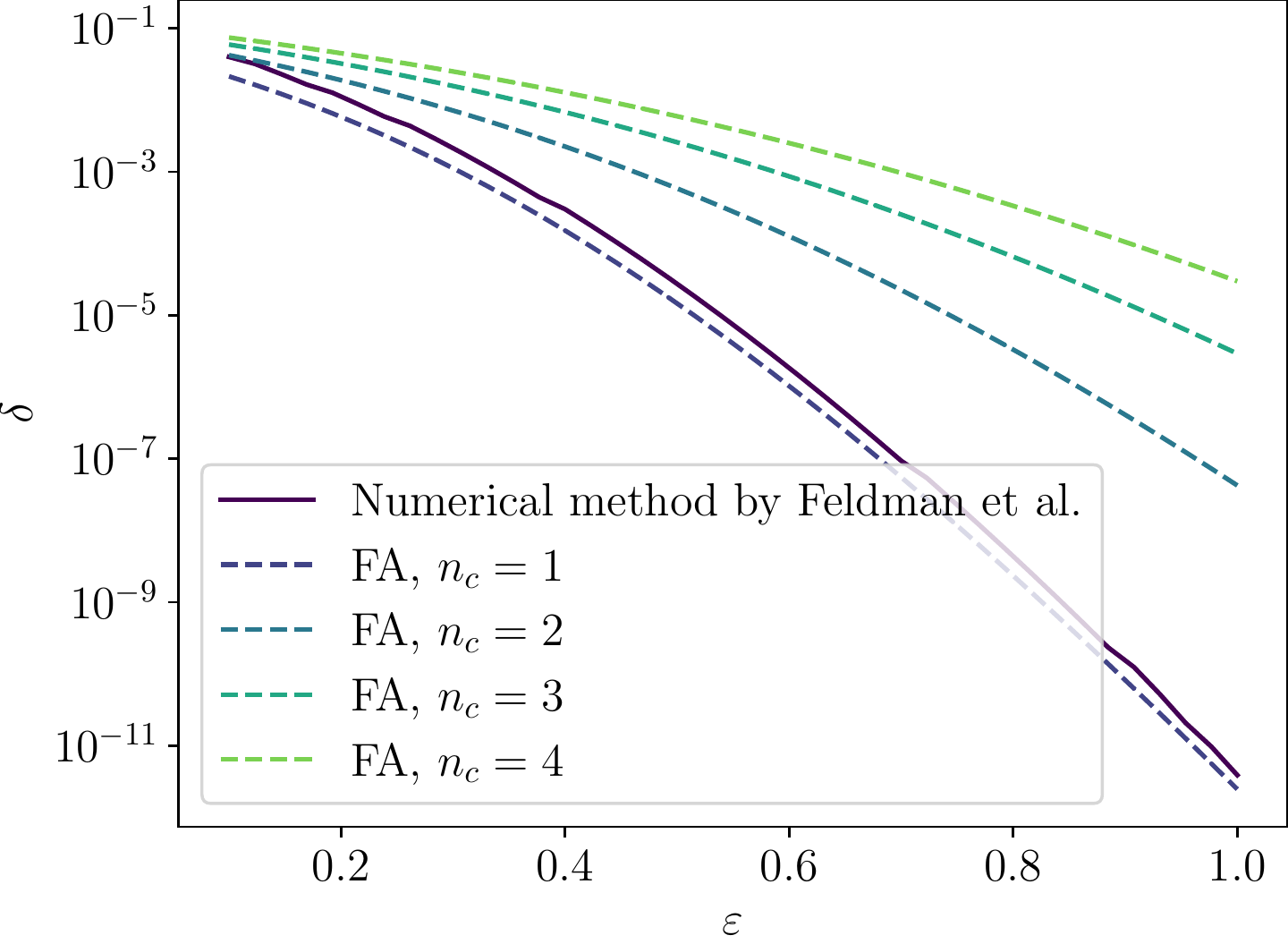}
        \caption{Evaluation of $\delta(\veps)$ for general single and composite shuffle $(\veps_0,0)$-LDP mechanisms: 
        for single composition protocols the numerical method by~\citet{feldman2021hiding} 
        is close to the tight bounds from FA ($n_c=1$). Their method is not directly applicable to compositions, for which the Fourier accountant also gives tight bounds. 
        Number of users $n=10^4$ and the LDP parameter $\veps_0=4.0$. To obtain the upper bounds using FA, we used parameter values $L=20$ and $m=10^7$.}
 	\label{fig:comparison}
\end{figure}

\subsection{Experimental comparison to the RDP bounds of~\citet{girgis2021shuffled}}
\label{sec:experiments_girgis}

\citet{girgis2021shuffled} consider a protocol where only 
a randomly sampled, fixed sized subset of users %that 
send contributions to the shuffler on each round. % is sampled randomly.
This can be seen as a composition of a shuffler and a subsampling mechanism. 
We can generalise our analysis to the subsampled case via 
%and this exact strategy is considered in 
Proposition 30 of~\citep{zhu2021optimal}, which states that if a pair of distributions $(P,Q)$ is a dominating pair of distributions
for a mechanism $\mathcal{M}$ for datasets of size $\gamma  n$ under $\sim$-neighbourhood relation (substitute relation), 
where $\gamma>0$ is the subsampling ratio (size of the subset divided by $n$), then
$(\gamma \cdot P + (1-\gamma) \cdot Q,Q)$ is a dominating distribution for the subsampled mechanism
$\mathcal{M} \circ S_{Subset}$, where the subsampling $S_{Subset}$ is carried out as described above.
By Lemma~\ref{lem:PQcomp} we know that the pair of distributions $(P,Q)$ of equation \eqref{eq:2n}, where 
$C \sim \mathrm{Bin}(\gamma n-1,\ee^{-\veps_0})$
give a dominating pair of distributions for a general $\veps_0$-LDP shuffler for datasets of size $\gamma n$,
and therefore we can obtain $(\veps,\delta)$-bounds for compositions of 
$\mathcal{M} \circ S_{Subset}$ using Corollary~\ref{cor:pld_compositions} and
the pair of distributions $(\gamma \cdot P + (1-\gamma) \cdot Q,Q)$.
As we see from Figure~\ref{fig:girgis1}, the PLD-based approach gives 
considerably lower $\veps(\delta)$-bounds. As $n_c$ increases, the FFT-based 
bound gets closer to the RDP bound, as noticed previously in~\citep{koskela2020} 
in the case of subsampled Gaussian mechanism.

\begin{figure} [h!]
     \centering
        \includegraphics[width=.6\textwidth]{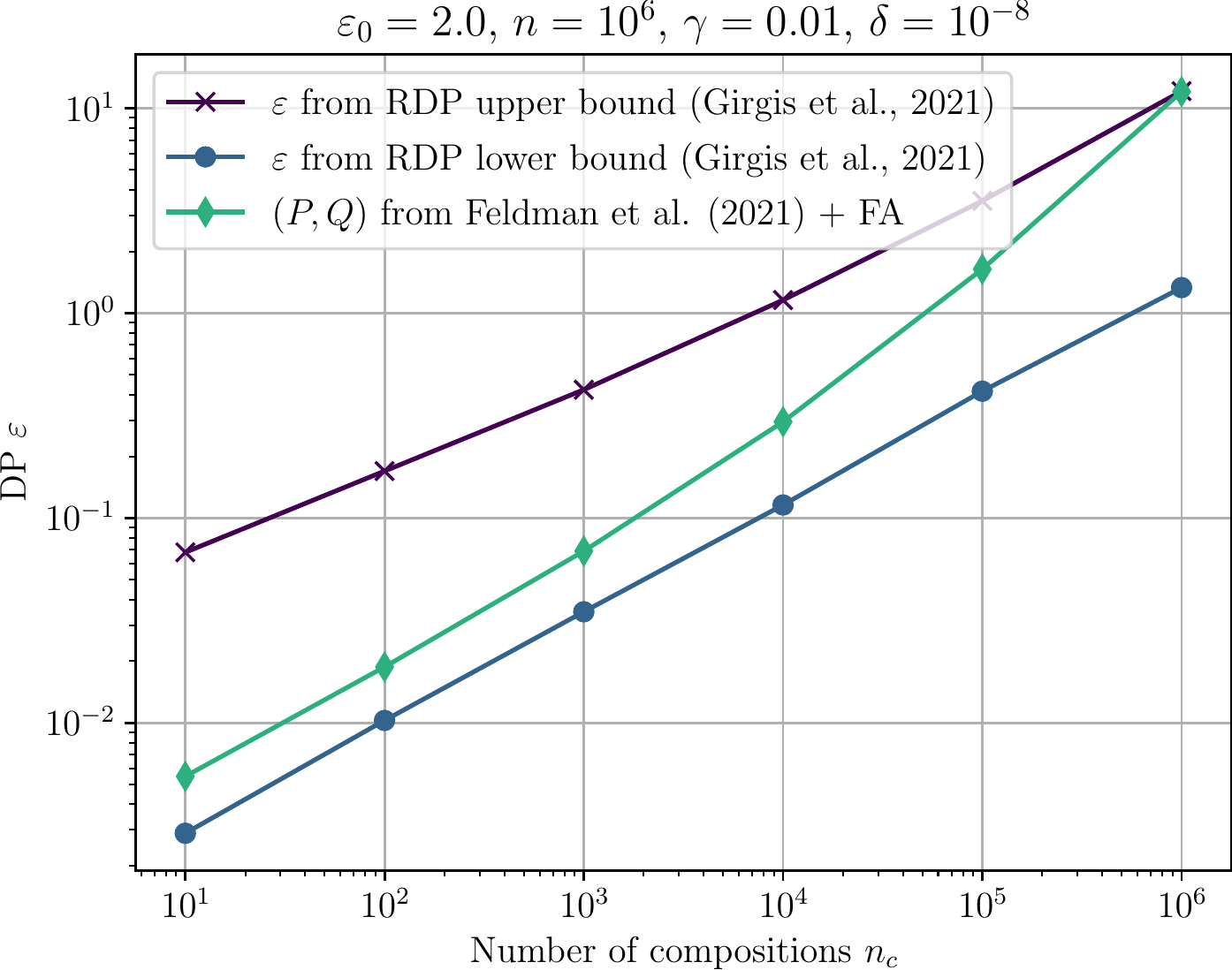} \\
        \vspace{3mm}
        \includegraphics[width=.6\textwidth]{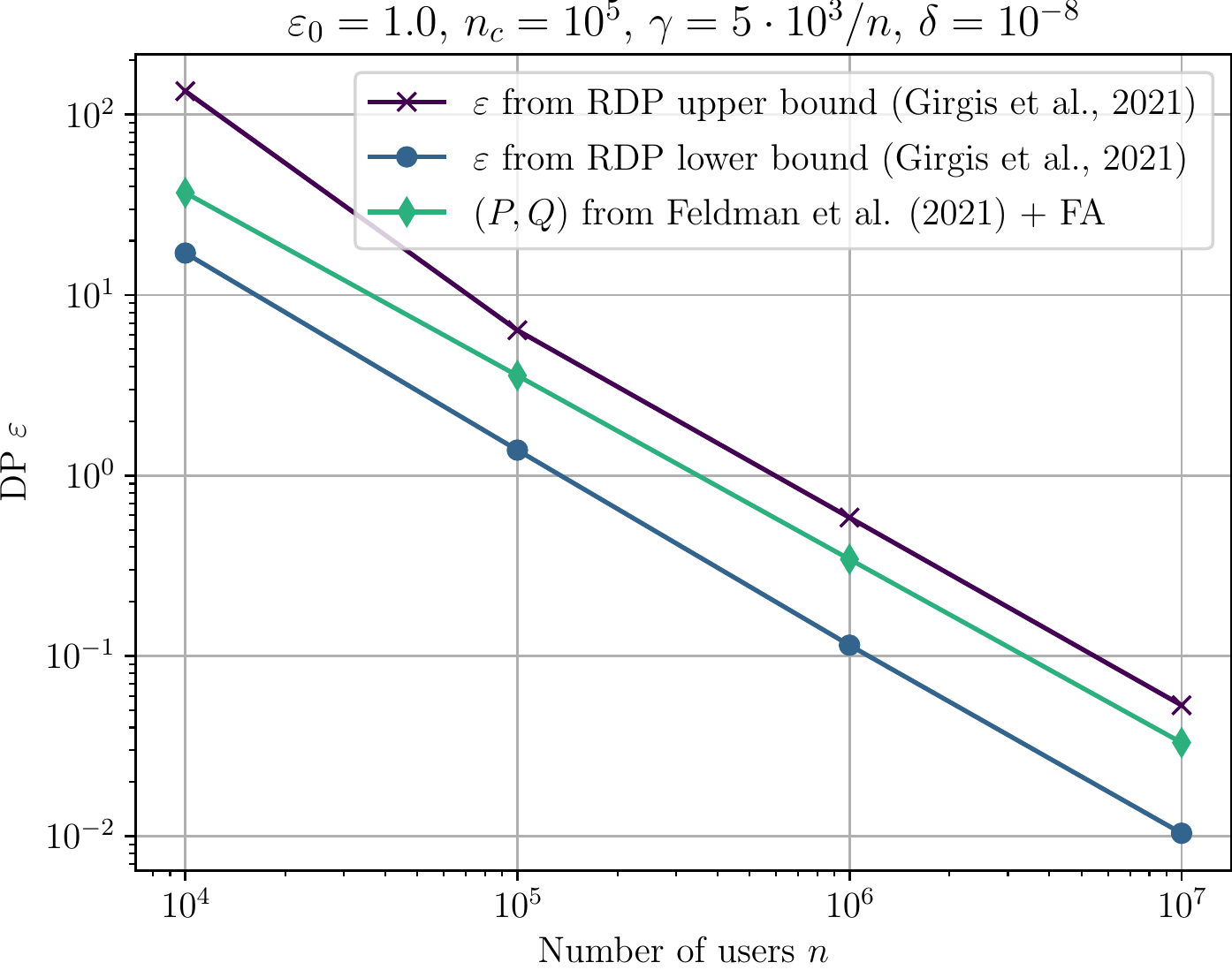}
        \caption{Evaluation of $\veps(\delta)$ for compositions of subsampled shufflers. 
        We compare the bounds obtained using FA and the PLD determined
        by the pair of distributions $(\gamma \cdot P + (1-\gamma) \cdot Q,Q)$ 
        ($P$ and $Q$ from of equation \eqref{eq:2n} with $n$ replaced by $\gamma n$)
        and the RDP-bounds given in Thm. 2 of~\citep{girgis2021shuffled}
        that are mapped to $\veps(\delta)$-bounds using Lemma 1 of~\citep{girgis2021shuffled}.
        % Results of~\citet{zhu2021optimal} guarantee that 
        % non-adaptive compositions of $(\gamma \cdot P + (1-\gamma) \cdot Q$ and $Q$ 
        Above: bounds for different numbers of users $n$ when number of
        compositions $n_c$ is fixed.
        Below: number of compositions $n_c$ varies and $n$ is fixed.
        Here $\gamma$ denotes the subsampling ratio.
  }
 	\label{fig:girgis1}
\end{figure}

%%%%%%%%%%%%%%%%%%%%%%%%%%%%%%%%%%%%%%%%%%%%%%%%%%%%%%%%%%%%%%%%%%%%%%%%%%%%%%%%%%%%%%%%%%%
%%%%%%%%%%%%%%%%%%%%%%%%%%%%%%%%%%%%%%%%%%%%%%%%%%%%%%%%%%%%%%%%%%%%%%%%%%%%%%%%%%%%%%%%%%%
%%%%%%%%%%%%%%%%%%%%%%%%%%%%%%%%%%%%%%%%%%%%%%%%%%%%%%%%%%%%%%%%%%%%%%%%%%%%%%%%%%%%%%%%%%%
\section{Shuffled $k$-randomised response}
\label{sec:kRR}
%%%%%%%%%%%%%%%%%%%%%%%%%%%%%%%%%%%%%%%%%%%%%%%%%%%%%%%%%%%%%%%%%%%%%%%%%%%%%%%%%%%%%%%%%%%
%%%%%%%%%%%%%%%%%%%%%%%%%%%%%%%%%%%%%%%%%%%%%%%%%%%%%%%%%%%%%%%%%%%%%%%%%%%%%%%%%%%%%%%%%%%
%%%%%%%%%%%%%%%%%%%%%%%%%%%%%%%%%%%%%%%%%%%%%%%%%%%%%%%%%%%%%%%%%%%%%%%%%%%%%%%%%%%%%%%%%%%

\citet{balle2019blanket} give a protocol for $n$ parties to compute a private histogram over the 
domain $[k]$ in the single-message shuffle model. 
%The local randomiser of this protocol is shown in Algorithm 1 of~\citep{balle2019blanket}.
The randomiser is parameterised by a probability $\gamma$, and consists of a $k$-ary randomised response mechanism ($k$-RR) that returns the true value  with probability $1 - \gamma$ and a uniformly random value with probability $\gamma$. 
%As noted by~\citet{balle2019blanket}, this randomiser has been considered before in the works XYZ 
%
%In~\citep{balle2019blanket}, the local randomiser (as described above) is denoted by the mechanism $\mathcal{R}_{\gamma,k,n}^{PH}$, 
Denote this $k$-RR randomiser by $\mathcal{R}_{\gamma,k,n}^{PH} $
and the shuffling operation by $\mathcal{S}$. Thus, we are studying the privacy
of the shuffled randomiser $\mathcal{M} = \mathcal{S} \circ \mathcal{R}_{\gamma,k,n}^{PH}$.

Consider first the proof of~\citet[Thm.\;3.1]{balle2019blanket}.
Assuming without loss of generality that the differing data element between $X$ and $X'$, $X,X' \in [k]^n$, is $x_n$, 
% and writing $\text{View}_{\mathcal M}^{A}(X)$ for the view of an adversary $A$ when $\mathcal M$ is run on data set $X$, 
the (strong) adversary $A_s$ used by \citet[Thm.\;3.1]{balle2019blanket} is defined as follows:
\begin{defn} \label{def:A_s}
Let $\mathcal M = \mathcal{S} \circ \mathcal{R}_{\gamma,k,n}^{PH}$ be the shuffled $k$-RR mechanism, and w.l.o.g. let the differing element be $x_n$. 
We define adversary $A_s$ as an adversary with the view
\begin{equation*}
	\begin{aligned}
		\viewS{X} = 
		\left((x_1,\dots,x_{n-1}), \quad \beta \in \{0,1\}^{n}, \quad (y_{\pi (1)},\dots, y_{\pi (n)}) \right),
	\end{aligned}
\end{equation*} 
where $\beta$ is a binary vector identifying which parties answered randomly, and $\pi$ is a uniformly random permutation applied by the shuffler.
\end{defn}
Assuming w.l.o.g. that the differing element $x_n=1$ and $x_n'=2$, the proof then shows that for any possible view $V$ of the adversary $A_s$,
$\tfrac{\mathbb{P}(\viewS{X} = V)}{\mathbb{P}(\viewS{X'} = V)  } = \tfrac{n_1}{n_2}$,
where $n_i$ denotes the number of messages received by the server with value 
$i$ after removing from the output $Y$ any truthful answers submitted by the first $n-1$ users. 
Moreover,~\citet{balle2019blanket} show that for all neighbouring $X$ and $X'$,
\begin{equation} \label{eq:viewsX}
    \viewS{X} \simeq_{(\veps,\delta)} \viewS{X'} 
\end{equation}
for
\begin{equation} \label{eq:PsQsX}
    \delta(\veps) = \mathbb{P}\left( \frac{N_1}{N_2} \geq \ee^\veps \right),
\end{equation}
where %the corresponding random variables 
%$N_1 \sim P_s$ and $N_2 \sim Q_s$,
%where 
\small
\begin{equation}
\label{eq:kRR_P_Q_nontight}
\begin{aligned}
N_1 \sim \mathrm{Bin}\left(n-1,\frac{\gamma}{k} \right) + 1, \quad 
N_2 \sim \mathrm{Bin}\left(n-1,\frac{\gamma}{k} \right).
\end{aligned}
\end{equation}
\normalsize
From the proof of~\citet[Thm.\;3.1]{balle2019blanket} we directly get the following result 
for adaptive compositions of the $k$-RR shuffler.

\begin{thm} \label{thm:comp_strong}

Consider $n_c$ adaptive compositions of the 
$k$-RR shuffler mechanism $\mathcal{M}$ and an adversary $A_s$
as described in Def.~\ref{def:A_s} above. Then, the tight $(\veps,\delta)$-bound is given by
$$
 \delta(\veps) = \mathbb{P} \left( \sum\limits_{i=1}^{n_c} Z_i \geq \veps   \right),
$$
where $Z_i$'s are independent and for all $1 \leq i \leq n_c $, 
$Z_i \sim \log\left( \frac{N_1}{N_2}  \right)$, where $N_1$ and $N_2$ are distributed as in \eqref{eq:kRR_P_Q_nontight}.

% $$ 
% N_1 \sim \mathrm{Bin}\left(n-1, \frac{\gamma}{k}\right) + 1,
% \quad N_2 \sim \mathrm{Bin}\left(n-1, \frac{\gamma}{k}\right).
% $$
\begin{proof}
We first remark that in fact \eqref{eq:PsQsX} holds when $\ee^\veps$ is replaced by any $\alpha \geq 0$,
i.e., for any neighbouring $X$ and $X'$, when $\alpha \geq 0$,
\begin{equation} \label{eq:HN1N2}
    H_\alpha\big(\viewS{X} || \viewS{X'}\big) = \mathbb{P}\left( \frac{N_1}{N_2} \geq \alpha \right),
\end{equation}
where $N_1 \sim \mathrm{Bin}(n-1, \frac{\gamma}{k}) + 1$, $N_2 \sim \mathrm{Bin}(n-1, \frac{\gamma}{k})$. 
This can be seen directly from the arguments of the proof of~\citet[Thm.\;3.1]{balle2019blanket}.
Next, we may use a similar argument as in the proof of~\citep[Thm.\;10]{zhu2021optimal}. By using \eqref{eq:HN1N2} repeatedly, we see 
that for an adaptive composition of two mechanisms $\mathcal{M}_1$ and $\mathcal{M}_2$:
%for obtaining an $\delta(\veps)$-upper bound for adaptive compositions:
\begin{equation*}
\begin{aligned}
  \delta(\veps) & = \mathbb{P}_{V \sim \text{View}_{\mathcal{M}_1}^{A_s}(X), V' \sim \text{View}_{\mathcal{M}_2}^{A_s}(X,V) }
  \left[ \frac{\mathbb{P}(\text{View}_{\mathcal{M}_1}^{A_s}(X)=V) \cdot \mathbb{P}(\text{View}_{\mathcal{M}_2}^{A_s}(X,V)=V')  }{\mathbb{P}(\text{View}_{\mathcal{M}_1}^{A_s}(X')=V) \cdot \mathbb{P}(\text{View}_{\mathcal{M}_2}^{A_s}(X',V)=V') } \geq \ee^\veps  \right] \\
  & = \mathbb{P}_{V \sim \text{View}_{\mathcal{M}_1}^{A_s}(X), V' \sim \text{View}_{\mathcal{M}_2}^{A_s}(X,V) }
  \left[ \frac{ \mathbb{P}(\text{View}_{\mathcal{M}_2}^{A_s}(X,V)=V')  }{ \mathbb{P}(\text{View}_{\mathcal{M}_2}^{A_s}(X',V)=V') } \geq \ee^{\veps - \log \frac{\mathbb{P}(\text{View}_{\mathcal{M}_1}^{A_s}(X)=V)}{\mathbb{P}(\text{View}_{\mathcal{M}_1}^{A_s}(X')=V)}}  \right] \\
  & = \mathbb{P}_{V \sim \text{View}_{\mathcal{M}_1}^{A_s}(X)} 
  \left[ \frac{ N_1^2 }{ N_2^2 } \geq \ee^{\veps - \log \frac{\mathbb{P}(\text{View}_{\mathcal{M}_1}^{A_s}(X)=V)}{\mathbb{P}(\text{View}_{\mathcal{M}_1}^{A_s}(X')=V)}}  \right] \\
  & = \mathbb{P}_{V \sim \text{View}_{\mathcal{M}_1}^{A_s}(X)} 
  \left[ \frac{\mathbb{P}(\text{View}_{\mathcal{M}_1}^{A_s}(X)=V)}{\mathbb{P}(\text{View}_{\mathcal{M}_1}^{A_s}(X')=V)} \geq \ee^{\veps - \log \frac{ N_1^2 }{ N_2^2 } }   \right] \\
   & = \mathbb{P}
  \left[ \frac{ N_1^1 }{ N_2^1 } \geq \ee^{\veps - \log \frac{ N_1^2 }{ N_2^2 } }   \right] \\
  & = \mathbb{P}
  \left[ \frac{ N_1^1 \cdot N_1^2 }{ N_2^1 \cdot N_2^2 } \geq \ee^{\veps }   \right] \\
  & = \mathbb{P} \left[ \log\left( \frac{N_1^1}{N_2^1} \right) +  \log\left( \frac{N_1^2}{N_2^2} \right) \geq \veps \right],
\end{aligned}
\end{equation*}
where $N_1^1, N_1^2 \sim \mathrm{Bin}(n-1, \frac{\gamma}{k}) + 1$, $N_2^1, N_2^2 \sim \mathrm{Bin}(n-1, \frac{\gamma}{k})$. The proof for $n_c>2$ goes analogously.\end{proof}
\end{thm}

%
% Thus, $\viewS{X} \simeq_{(\veps,\delta)} \viewS{X'} $ if 
% $P_{s} \simeq_{(\veps,\delta)} Q_s $.
\citet{balle2019blanket} showed that for adversary $A_s$ 
the shuffled mechanism $\mathcal{M} = \mathcal{S} \circ \mathcal{R}_{\gamma,k,n}^{PH}$ is $(\veps,\delta)$-DP for any $k,n \in \mathbb{N}$, $\veps \leq 1$ and $\delta \in (0,1]$ 
such that
$\gamma = \max \left\{ \tfrac{14 \cdot k \cdot \log\left( 2/\delta   \right)}{ (n-1) \cdot \veps^2}, \tfrac{27 \cdot k}{(n-1) \cdot \veps}    \right\}.$
Comparison to this bound is shown in Figure~\ref{fig:blanket1}.

\subsection{Tight bounds for varying adversaries using Fourier accountant}
\label{sec:krr_FA}

Following the reasoning of the proof of~\citet[Thm.\;3.1]{balle2019blanket}, 
for adversary $A_{s}$ (see Def.~\ref{def:A_s}),
 we can compute tight $\delta(\veps)$-bounds using Thm.~\ref{thm:comp_strong}. 

Having tight bounds also enables us to evaluate 
exactly how much 
different assumptions on the adversary cost us in terms of privacy. 
For example, instead of the adversary $A_s$ 
%, who has perfect knowledge of the set of parties who answered honestly, 
we can analyse a weaker adversary $A_w$, 
who has extra information only on the first $n-1$ parties. 
We formalise this as follows:
\begin{defn} \label{def:A_w}
Let $\mathcal M = \mathcal{S} \circ \mathcal{R}_{\gamma,k,n}^{PH}$ be the shuffled $k$-RR mechanism, and w.l.o.g. let the differing element be $x_n$. 
Adversary $A_w$ is an adversary with the view
\begin{equation*}
	\begin{aligned}
		\view{X} = \left((x_1,\dots,x_{n-1}), \quad \beta \in \{0,1\}^{n-1}, \quad (y_{\pi (1)},\dots, y_{\pi (n)}) \right),
	\end{aligned}
\end{equation*} 
where $\beta$ is a binary vector identifying which of the first $n-1$ parties answered randomly, and $\pi$ is a uniformly random permutation applied by the shuffler.
\end{defn}

Note that compared to the stronger adversary $A_s$ formalised in Def.~\ref{def:A_s} 
the difference is only in the vector $\beta$. 
We write $b = \sum_i \beta_i$, and $B$ for the corresponding random variable in the following.

%Assume w.l.o.g. that the neighbouring (worst-case) data sets are $X = (x_1,\ldots,x_{n-1},1)$ and $X' = (x_1,\ldots,x_{n-1},2)$. 
The next theorem gives the random variables we need to calculate 
privacy bounds for adversary $A_w$:
\begin{thm} \label{thm:P_w,Q_w}
Assume w.l.o.g. differing elements $x_n=1, x_n'=2$, and 
adversary $A_w$ as given in Def.~\ref{def:A_w}. 
To find a tight DP bound for $\mathcal M= \mathcal{S} \circ \mathcal{R}_{\gamma,k,n}^{PH}$ we can equivalently analyse 
the random variables $P_w, Q_w$ defined as 
\begin{equation} \label{eq:Pw_Qw}
P_w = P_1 + P_2, \quad Q_w = Q_1 + Q_2,   
\end{equation}
where 
\begin{equation*}
	\begin{aligned}
		&P_1 \sim (1-\gamma) \cdot  N_1|B, \quad P_2 \sim  \frac{\gamma}{k} \cdot (B+1), \\
		&Q_1 \sim  (1-\gamma) \cdot N_2|B, \quad Q_2 \sim  \frac{\gamma}{k} \cdot (B+1),
	\end{aligned}
\end{equation*}
\begin{equation*} %\label{eq:k_P_Q_2}
\begin{aligned}
& B \sim \mathrm{Bin} (n-1, \gamma), \\
%\text{  } 
&N_i^B|B \sim \mathrm{Bin}(B, 1/k ), \quad i=1,\dots,k, \\
%\begin{aligned}
&N_1|B = N_1^B|B + \mathrm{Bern}(1-\gamma + \gamma/k) \\
% = \mathrm{Bin} (n-1, \gamma/k) + \mathrm{Bern}(1-\gamma + \gamma/k),   \\
&N_2|B = N_2^B|B  + \mathrm{Bern}(\gamma/k).
\end{aligned}
\end{equation*}
\end{thm}

As a direct corollary to this theorem, and analogously to Thm.~\ref{thm:comp_strong}, we have the following
result which allows computing tight $\delta(\veps)$-bounds against the adversary $A_w$ for adaptive compositions.
\begin{thm} \label{thm:comp_weak}

Consider $n_c$ adaptive compositions of the $k$-RR shuffler mechanism $\mathcal{M}$ and an adversary $A_w$
as described in Def.~\ref{def:A_w} above. Then, the tight $(\veps,\delta)$-bound is given by
$$
 \delta(\veps) = \mathbb{P} \left( \sum\limits_{i=1}^{n_c} Z_i \geq \veps   \right),
$$
where $Z_i$'s are independent and for all $1 \leq i \leq m$, 
$$
Z_i \sim \log\left( \frac{N_1}{N_2}  \right), \quad 
N_1 \sim P_w,
\quad N_2 \sim Q_w,
$$
where $P_w$ and $Q_w$ are given in \eqref{eq:Pw_Qw}.
\end{thm}
\begin{proof}
See Thm.~\ref{thm:comp_strong} proof.
\end{proof}

Figure~\ref{fig:blanket1} shows an empirical comparison of 
the tight bounds obtained with Fourier accountant assuming 
the stronger adversary $A_s$, which leads to the neighbouring random variables 
$P_s, Q_s$ from \eqref{eq:kRR_P_Q_nontight}, 
or the weaker adversary $A_w$, corresponding to $P_w,Q_w$ from 
Thm.~\ref{thm:P_w,Q_w}, together with the loose analytic bounds from \citet[Thm.\;3.1]{balle2019blanket}. As shown in the Figure, tight 
bounds are considerably tighter than the analytic one. 
There is also a clear difference in the tight bounds 
resulting from assuming either 
the strong adversary $A_s$ or the weaker $A_w$. 
We remark that the evaluation of the distributions for $Z_i$'s in 
theorems~\ref{thm:comp_strong} and~\ref{thm:comp_weak} can be carried out
in high accuracy in $\mathcal{O}(n)$-time using Hoeffding's inequality similarly as in Lemma~\ref{sec:clones_eff_approximation}.

\begin{figure}[h!]
     \centering
     %\begin{subfigure}[b]{0.48\textwidth}
         \centering
         \includegraphics[width=.6\textwidth]{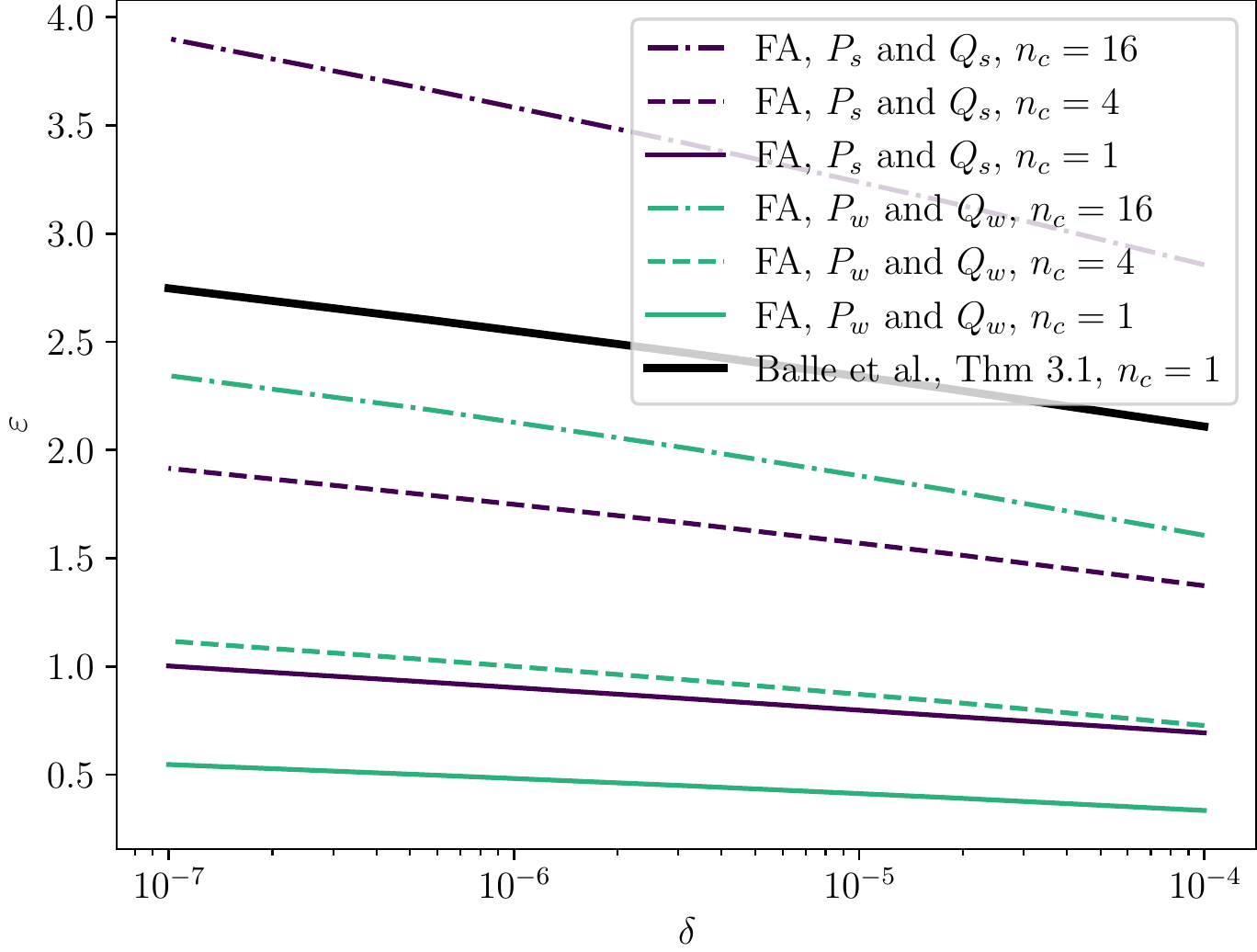}
     \caption{Shuffled $k$-randomised response: tight bounds are  significantly better than the existing analytic one. 
     Tight $(\veps,\delta)$-DP bounds obtained using the Fourier accountant (FA) for different number of compositions $n_c$, and the loose analytical bound from  \citet[Thm.\;3.1]{balle2019blanket} for a single composition. We apply FA to the $\delta(\veps)$-expression of Thm.~\ref{thm:comp_strong} ($P_s$ and $Q_s$),
     and to the $\delta(\veps)$-expression of 
     Thm.~\ref{thm:comp_weak} ($P_w$ and $Q_w$).
     Both are tight bounds under the assumed adversary (stronger and weaker). FA with $P_s,Q_s$ and $n_c=1$ is the tight bound 
     with the same assumptions as used in the loose analytic bound. 
     Total number of users $n=1000$,  probability of randomising for each user $\gamma=0.25$, and $k=4$. 
     For FA, we use parameter values $L=20$ and $m=10^7$.
     \label{fig:blanket1}}
\end{figure}

\section{On the difficulty of obtaining bounds in the general case} \label{sec:diff}

We have provided means to compute accurate $(\veps,\delta)$-bounds for the general $\veps_0$-LDP shuffler using the results by~\citet{feldman2021hiding}  
and tight bounds for the case of $k$-randomised response.
Using the following example, we illustrate the computational difficulty of obtaining tight bounds for arbitrary local randomisers.
Consider neighbouring datasets $X,X' \in \mathbb{R}^n$, where all elements of $X$ are equal, and $X'$ contains one element differing by 1.
Without loss of generality (due to shifting and scaling invariance of DP), we may consider the case where $X$ consists of zeros and $X'$
has 1 at some element. Considering a mechanism $\mathcal{M}$ that consists of adding Gaussian noise with variance $\sigma^2$ to each element and then shuffling,
we see that the adversary sees the output of $\mathcal{M}(X)$ distributed as 
$$
\mathcal{M}(X) \sim \mathcal{N}(0,\sigma^2 I_n), 
$$
and the output
$\mathcal{M}(X')$ as the mixture distribution 
$$
\mathcal{M}(X') \sim \tfrac{1}{n} \cdot \mathcal{N}(e_1,\sigma^2 I_n) + \ldots + \tfrac{1}{n} \cdot \mathcal{N}(e_n,\sigma^2 I_n),
$$
where $e_i$ denotes the $i$th unit vector. %, i.e. $e_1 = \begin{bmatrix} 1 & 0 & \ldots & 0 \end{bmatrix}$ and so on.
Determining the hockey-stick divergence 
$H_{\ee^\veps}(\mathcal{M}(X') || \mathcal{M}(X))$ cannot be projected to a lower-dimensional problem, unlike in the case of the (subsampled) Gaussian mechanism,
for example, which is equivalent to a one-dimensional problem~\cite{koskela2021heterogeneous}. This means that in order to obtain tight $(\veps,\delta)$-bounds,
we need to numerically evaluate the $n$-dimensional hockey-stick integral $H_{\ee^\veps}(\mathcal{M}(X') || \mathcal{M}(X))$. % (eq. \eqref{eq:alphadiv}).
Using a numerical grid as in FFT-based accountants is unthinkable due to the curse of the dimensionality. 
However, we may use the fact that for any data set $X$, the density function $f_X(t)$ of $\mathcal{M}(X)$ is a permutation-invariant function,
meaning that for any $t \in \mathbb{R}^n$ and for any permutation $ \sigma \in \pi_n$, $f_X\big(\sigma(t)\big) = f_X(t)$.
This allows reduce the number of required points on a regular grid for the hockey stick integral from $O(m^n)$ to $O(m^n/n!)$,
where $m$ is the number of discretisation points in each dimension.
Recent research on numerical integration of permutation-invariant functions \citep[e.g.][]{Nuyens2016Rank} suggests it may be possible to significantly reduce or even eliminate the dependence on $n$ using more advanced integration techniques.
In Figure~\ref{fig:mc} we have computed
$H_{\ee^\veps}(\mathcal{M}(X') || \mathcal{M}(X))$ up to $n=7$ using Monte Carlo integration on a hypercube $[-L,L]^n$ which requires $\approx 5 \cdot 10^7$ samples for getting two correct significant figures for $n=7$.

\begin{figure}[h!]
     \centering
        \includegraphics[width=.6\textwidth]{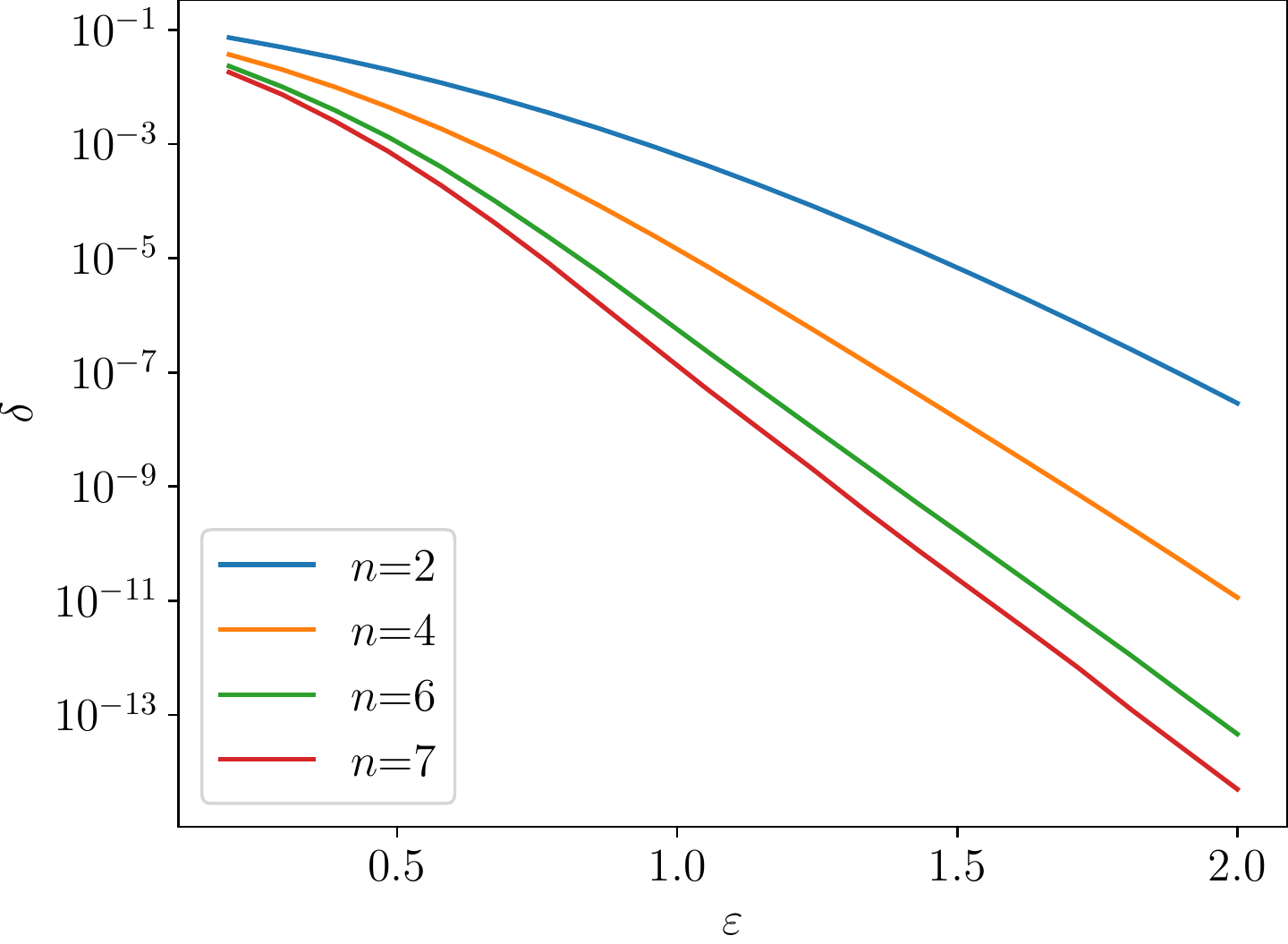}
        \caption{Approximation of tight $\delta(\veps)$ for shuffled outputs of Gaussian mechanisms ($\sigma=2.0$) by Monte Carlo integration of the 
		hockey-stick divergence $H_{\ee^\veps}(\mathcal{M}(X') || \mathcal{M}(X))$, using $5 \cdot 10^7$ samples (two correct significant figures).
}
 	\label{fig:mc}
\end{figure}

\section{Discussion}

We have shown how numerical privacy accounting can be used to calculate accurate upper 
bounds for compositions of various $(\veps,\delta)$-DP mechanisms and different adversaries in the shuffle model. 
An alternative approach %for tightening the accounting of the compositions 
would be to use the R\'enyi differential privacy~\citep{mironov2017}. 
However, as illustrated by the comparison against the results of \citet{girgis2021shuffled} in Fig. \ref{fig:girgis1}, our numerical method leads to considerably tighter bounds.
For shuffled mechanisms, the difference appears even more significant than for regular DP-SGD \citep{koskela2020,koskela2021tight}, showing up to an order of magnitude reduction in $\veps$.

Numerical and analytical privacy bounds are in many cases complementary and serve different purposes.
Numerical accountants allow finding the tightest possible bounds for production and enable more unbiased comparison of algorithms when accuracy of accounting is not a factor.
Analytical bounds enable theoretical research and understanding of scaling properties of algorithms, but the inaccuracy of the bounds raises the risk of misleading conclusions about privacy claims.

While our results provide significant improvements over previous state-of-the-art, they only provide optimal accounting for $k$-randomised response.
Developing optimal accounting for more general mechanisms as well as
extending the results to $(\veps_0, \delta_0)$-LDP base mechanisms are important topics for future research.

%\note{Should we say something about $(\veps, delta)$-LDP base mechanisms? Feldman et al. claim some results, and the last Gaussian example would need this.}

%\note{check refs to supplement/appendix are consistent}

%%%%%%%%%%%%%%%%%%%%%%%%%%%%%%%%%%%%%%%%%%%%%%%%%%%%%%%%%%%%%%%%%%%%%%%%%%%%%%%%%%%%%%%%%%%%%%%%%%%%%%%%%%%%%%%%%%%%%%%%%%
\section*{Acknowledgements} 
%%%%%%%%%%%%%%%%%%%%%%%%%%%%%%%%%%%%%%%%%%%%%%%%%%%%%%%%%%%%%%%%%%%%%%%%%%%%%%%%%%%%%%%%%%%%%%%%%%%%%%%%%%%%%%%%%%%%%%%%%%

This work has been supported by the Academy of Finland [Finnish Center for Artificial Intelligence FCAI and grant 325573] and by the Strategic Research Council at the Academy of Finland [grant 336032].

\bibliography{pld}

\begin{thebibliography}{}

\bibitem[Balle et~al., 2019]{balle2019blanket}
Balle, B., Bell, J., Gasc{\'o}n, A., and Nissim, K. (2019).
\newblock The privacy blanket of the shuffle model.
\newblock In {\em Annual International Cryptology Conference}, pages 638--667.
  Springer.

\bibitem[Balle et~al., 2020]{balle2020multi}
Balle, B., Bell, J., Gasc{\'o}n, A., and Nissim, K. (2020).
\newblock Private summation in the multi-message shuffle model.
\newblock In {\em Proceedings of the 2020 ACM SIGSAC Conference on Computer and
  Communications Security}, pages 657--676.

\bibitem[Bittau et~al., 2017]{Bittau2017}
Bittau, A., Erlingsson, {\'U}., Maniatis, P., Mironov, I., Raghunathan, A.,
  Lie, D., Rudominer, M., Kode, U., Tinnes, J., and Seefeld, B. (2017).
\newblock Prochlo: Strong privacy for analytics in the crowd.
\newblock In {\em Proceedings of the 26th Symposium on Operating Systems
  Principles}, pages 441--459.

\bibitem[Cheu et~al., 2019]{cheu2019distributed}
Cheu, A., Smith, A., Ullman, J., Zeber, D., and Zhilyaev, M. (2019).
\newblock Distributed differential privacy via shuffling.
\newblock In {\em Annual International Conference on the Theory and
  Applications of Cryptographic Techniques}, pages 375--403. Springer.

\bibitem[Dwork et~al., 2006]{dwork_et_al_2006}
Dwork, C., McSherry, F., Nissim, K., and Smith, A. (2006).
\newblock Calibrating noise to sensitivity in private data analysis.
\newblock In {\em Proc. TCC 2006}, pages 265--284.

\bibitem[Dwork and Roth, 2014]{DworkRoth}
Dwork, C. and Roth, A. (2014).
\newblock The algorithmic foundations of differential privacy.
\newblock {\em Found. Trends Theor. Comput. Sci.}, 9(3--4):211--407.

\bibitem[Erlingsson et~al., 2019]{Erlingsson2019}
Erlingsson, {\'U}., Feldman, V., Mironov, I., Raghunathan, A., Talwar, K., and
  Thakurta, A. (2019).
\newblock Amplification by shuffling: From local to central differential
  privacy via anonymity.
\newblock In {\em Proceedings of the Thirtieth Annual ACM-SIAM Symposium on
  Discrete Algorithms}, pages 2468--2479. SIAM.

\bibitem[Feldman et~al., 2021]{feldman2021hiding}
Feldman, V., McMillan, A., and Talwar, K. (2021).
\newblock Hiding among the clones: A simple and nearly optimal analysis of
  privacy amplification by shuffling.
\newblock In {\em 2021 IEEE 62nd Annual Symposium on Foundations of Computer
  Science}. IEEE.

\bibitem[Ghazi et~al., 2021]{ghazi2021}
Ghazi, B., Golowich, N., Kumar, R., Pagh, R., and Velingker, A. (2021).
\newblock On the power of multiple anonymous messages: Frequency estimation
  and selection in the shuffle model of differential privacy.
\newblock In Canteaut, A. and Standaert, F.-X., editors, {\em Advances in
  Cryptology -- EUROCRYPT 2021}, pages 463--488, Cham. Springer International
  Publishing.

\bibitem[Girgis et~al., 2021]{girgis2021shuffled}
Girgis, A., Data, D., Diggavi, S., Kairouz, P., and Suresh, A.~T. (2021).
\newblock Shuffled model of differential privacy in federated learning.
\newblock In {\em International Conference on Artificial Intelligence and
  Statistics}, pages 2521--2529. PMLR.

\bibitem[Gopi et~al., 2021]{gopi2021}
Gopi, S., Lee, Y.~T., and Wutschitz, L. (2021).
\newblock Numerical composition of differential privacy.
\newblock In {\em Advances in Neural Information Processing Systems}.

\bibitem[Kasiviswanathan et~al., 2011]{kasiviswanathan2011}
Kasiviswanathan, S.~P., Lee, H.~K., Nissim, K., Raskhodnikova, S., and Smith,
  A. (2011).
\newblock What can we learn privately?
\newblock {\em SIAM Journal on Computing}, 40(3):793--826.

\bibitem[Koskela and Honkela, 2021]{koskela2021heterogeneous}
Koskela, A. and Honkela, A. (2021).
\newblock Computing differential privacy guarantees for heterogeneous
  compositions using fft.
\newblock {\em arXiv preprint arXiv:2102.12412}.

\bibitem[Koskela et~al., 2020]{koskela2020}
Koskela, A., J{\"a}lk{\"o}, J., and Honkela, A. (2020).
\newblock Computing tight differential privacy guarantees using {FFT}.
\newblock In {\em International Conference on Artificial Intelligence and
  Statistics}, pages 2560--2569. PMLR.

\bibitem[Koskela et~al., 2021]{koskela2021tight}
Koskela, A., J{\"a}lk{\"o}, J., Prediger, L., and Honkela, A. (2021).
\newblock Tight differential privacy for discrete-valued mechanisms and for the
  subsampled gaussian mechanism using {FFT}.
\newblock In {\em International Conference on Artificial Intelligence and
  Statistics}, pages 3358--3366. PMLR.

\bibitem[Mironov, 2017]{mironov2017}
Mironov, I. (2017).
\newblock R\'enyi differential privacy.
\newblock In {\em 2017 IEEE 30th Computer Security Foundations Symposium
  (CSF)}, pages 263--275.

\bibitem[Nuyens et~al., 2016]{Nuyens2016Rank}
Nuyens, D., Suryanarayana, G., and Weimar, M. (2016).
\newblock Rank-1 lattice rules for multivariate integration in spaces of
  permutation-invariant functions - error bounds and tractability.
\newblock {\em Adv. Comput. Math.}, 42(1):55--84.

\bibitem[Sommer et~al., 2019]{sommer2019privacy}
Sommer, D.~M., Meiser, S., and Mohammadi, E. (2019).
\newblock Privacy loss classes: The central limit theorem in differential
  privacy.
\newblock {\em Proceedings on Privacy Enhancing Technologies},
  2019(2):245--269.

\bibitem[Zhu et~al., 2021]{zhu2021optimal}
Zhu, Y., Dong, J., and Wang, Y.-X. (2021).
\newblock Optimal accounting of differential privacy via characteristic
  function.
\newblock {\em arXiv preprint arXiv:2106.08567}.

\end{thebibliography}
\bibliographystyle{icml2022}

\appendix
% \appendix
\onecolumn
\newpage

%%%%%%%%%%%%%%%%%%%%%%%%%%%%%%%%%%%%%%%%%%%%%%%%%%%%%%%%%%%%%%%%%
%%%%%%%%%%%%%%%%%%%%%%%%%%%%%%%%%%%%%%%%%%%%%%%%%%%%%%%%%%%%%%%%%
\section{Auxiliary results for Section~\ref{pld_for_shuffling}}
%%%%%%%%%%%%%%%%%%%%%%%%%%%%%%%%%%%%%%%%%%%%%%%%%%%%%%%%%%%%%%%%%
%%%%%%%%%%%%%%%%%%%%%%%%%%%%%%%%%%%%%%%%%%%%%%%%%%%%%%%%%%%%%%%%%

% We give the needed expressions to determine the PLD $\omega_{P/Q}$.

In this section we give the needed expressions to determine the PLD 
\begin{equation*} %\label{eq:clones_pld}
	\omega_{ P/Q }(s) = \sum\nolimits_{a,b} \mathbb{P}(P=(a,b)) \cdot \delta_{s_{a,b}}(s),
\end{equation*}
where 
$$
s_{a,b} = \log \left( \frac{ \mathbb{P}(P=(a,b)) }{ \mathbb{P}(Q=(a,b)) }   \right).
$$
%is given by Lemma X.%~\ref{lem:logPQ}.
With these expressions, we can also determine the probability
\begin{equation*} %\label{eq:clones_delta_infty}
\begin{aligned}
\delta_{P/Q}(\infty) &= \sum_{ \{ (a,b) \, : \, \mathbb{P}( Q = (a,b)) = 0 ) \}} \mathbb{P}( P = (a,b) ).
% & = \left( \frac{1}{2} \right)^{n-1} \frac{ \ee^{-(n-2) \veps_0} }{\ee^{\veps_0}+1}.
\end{aligned}
\end{equation*}

Recall: denoting $q=\frac{\ee^{\veps_0}}{\ee^{\veps_0}+1}$, the distributions in \eqref{eq:2n} are given by the mixture distributions
\begin{equation} \label{Aeq:wtPQ_supp}
\begin{aligned}
P &= q \cdot P_1 + (1-q) \cdot P_0, \\ 
	Q &= (1-q) \cdot P_1 + q \cdot P_0,
\end{aligned}
\end{equation}
where
\begin{equation*}% \label{eq:P1P0Q1}
P_1 = (A+1,C-A), \quad P_0 = (A,C-A+1),
\end{equation*}
%Recall also that $q= \tfrac{\ee^{\veps_0}}{\ee^{\veps_0}+1}$ and
\begin{equation*}
	\begin{aligned}
		C \sim \mathrm{Bin}(n-1,\ee^{-\veps_0}), \quad A \sim \mathrm{Bin}(C,\tfrac{1}2).
	\end{aligned}
\end{equation*}

\subsection{Determining the log ratios $s_{a,b}$}

To determine $s_{a,b}$'s, we need the following auxiliary results.
\begin{lem} \label{Alem:P1P0}
When $b>0$ and $a>0$, we have:
$$
\mathbb{P}(P_1=(a,b)) = \frac{a}{b} \cdot \mathbb{P}(P_0=(a,b)).
$$
\begin{proof}
We see that $P_1=(a,b)$ if and only if $A=a-1$ and $C = a+b-1$. 
Since
\begin{equation*}
	\begin{aligned}
\mathbb{P}(A=a-1 \, | \,  C=a+b-1) 
&  = {a+b-1 \choose a-1} \frac{1}{2^{a+b-1}} \\
&  = \frac{a}{b} \cdot {a+b-1 \choose a} \frac{1}{2^{a+b-1}} \\
&   = \frac{a}{b} \cdot \mathbb{P}(A=a \, | \,  C=a+b-1),
	\end{aligned}
\end{equation*}
we see that
\begin{equation*}
	\begin{aligned}
		\mathbb{P}(P_1=(a,b))
		&  = \mathbb{P}(C=a+b-1) \cdot \mathbb{P}(A=a-1 \, | \,  C=a+b-1) \\
		& = \mathbb{P}(C=a+b-1) \cdot \frac{a}{b} \cdot \mathbb{P}(A=a \, | \,  C=a+b-1) \\
		& = \frac{a}{b} \cdot \mathbb{P}(P_0=(a,b)),
	\end{aligned}
\end{equation*}
since $P_0=(a,b)$ if and only if $A=a$ and $C = a+b-1$.
\end{proof}
\end{lem}

Using these expressions, and the fact that $\mathbb{P}(P_0=(a,0))=0$ for all $a$ and
$\mathbb{P}(P_1=(0,b))=0$ for all $b$, we get the following expressions needed for $s_{a,b}$'s.
\begin{lem} \label{Alem:logPQ}
When $b>0$ and $a \geq 0$,
\begin{equation*} % \label{Aeq:logwtPwtQ}
	  \frac{ \mathbb{P}(P=(a,b)) }{ \mathbb{P}(Q=(a,b)) } 
	= \frac{ q \cdot \frac{a}{b} + (1-q) }{q + (1-q)\frac{a}{b} }.
\end{equation*}
When $0 < a \leq n$,
$$
\frac{ \mathbb{P}(P=(a,0)) }{ \mathbb{P}(Q=(a,0)) } = \frac{q}{1-q}.
$$
\end{lem}

\subsection{Probabilities $\mathbb{P}(P=(a,b))$}

To determine $\omega_{P/Q}$, we still need to determine $\mathbb{P}(P=(a,b))$'s. These are given by
the following expressions.

\begin{lem} \label{Alem:P1P0_2}
When $a>0$,
\begin{equation*} % \label{eq:pab}
    \begin{aligned}
    \mathbb{P}(P_1=(a,b)) = {n-1 \choose i}  {i \choose j} p^i (1-p)^{n-1-i} \frac{1}{2^i},
    \end{aligned}
\end{equation*}
where $(a,b) = (j+1,i-j)$ (i.e., $C=i$ and $A=j$), and
\begin{equation*} % \label{eq:pab2}
 \mathbb{P}(P_0=(a,b))  = \frac{\ee^{-\veps_0}}{1 - \ee^{-\veps_0}} \frac{n-a-b}{2 a}  \mathbb{P}(P_1=(a,b)).
\end{equation*}
For $0 < b \leq n$,
$$
\mathbb{P}(P_1=(0,b))=0
$$
and
$$
\mathbb{P}(P_0=(0,b)) = {n-1 \choose b-1} \left( \frac{\ee^{-\veps_0}}{2}\right)^{b-1} (1-\ee^{-\veps_0})^{n-b}.
$$
\begin{proof}
The expressions follow directly from the definitions of $P_0$, $P_1$, $A$ and $C$.
\end{proof}
\end{lem}

\section{More detailed derivation of the probabilities for $k$-ary RR}

Recall from Section 5.1 of the main text: we consider the case where the adversary sees a vector 
$\beta$ of length $n-1$ identifying clients who submit only noise, 
except for the client with the differing element, and write $b=\sum_i \beta_i$. The adversary can 
remove all truthfully reported values by the clients $[n-1]$. 
Denote the observed counts after removal by $n_i, i=1,\dots,k$, so 
$\sum_{i=1}^k n_i = b+1$, and write $\mathcal R$ for the local randomiser. 
We now have
\begin{equation*}
	\begin{aligned}
	\mathbb P ( \view{ \mathbf x} = V )
	&= \sum_{i=1}^k 
	    \mathbb P (N_1=n_1, \dots, N_{i}=n_{i}-1,  N_{i+1}=n_{i+1}, \ldots \\
	    &  N_k=n_k | B ) \cdot
	    \mathbb P (\mathcal R (x_n) = i ) \cdot \mathbb P (B=b) \\
	 &= \binom{\B}{n_1-1,n_2,\dots,n_k}\left(\frac{1}{k}\right)^{\B} \cdot \left(1-\gamma + \frac{\gamma}{k}\right) 
	    \cdot \gamma^{\B}(1-\gamma)^{n-1-\B} \\
	    & + \sum_{i=2}^k \binom{\B}{n_1,\dots,n_i-1,n_{i+1},\dots,n_k}\left(\frac{1}{k}\right)^{\B} \cdot \frac{\gamma}{k} 
	    \cdot \gamma^{\B}(1-\gamma)^{n-1-\B} \\
	 &= \binom{\B}{n_1,n_2,\dots,n_k} \frac{ \gamma^{\B}(1-\gamma)^{n-1-\B} }{k^{\B}}  
	    \left[ n_1 (1-\gamma + \frac{\gamma}{k}) + \sum_{i=2}^k n_i \frac{\gamma}{k} \right] \\
	 &= \binom{\B}{n_1,n_2,\dots,n_k} \frac{ \gamma^{\B}(1-\gamma)^{n-1-\B} }{k^{\B}} \cdot \\ 
	& \quad \quad \quad    \left[ n_1 (1-\gamma + \frac{\gamma}{k}) + (\B+1 - n_1)  \frac{\gamma}{k}  \right] \\
	 &= \binom{\B}{n_1,n_2,\dots,n_k} \frac{ \gamma^{\B}(1-\gamma)^{n-1-\B} }{k^{\B}}  
	    \left[ n_1 (1-\gamma) + \frac{\gamma}{k} (\B+1) \right].
	\end{aligned}
\end{equation*}
Noting that $ \mathbb P (\mathcal R( x_n')=i) = (1-\gamma+\frac{\gamma}{k})$ when $i=2$ and $\frac{\gamma}{k}$ otherwise, 
repeating essentially the above steps gives
\begin{align*}
\mathbb P ( \view{ \mathbf x'} = V ) 
&= \sum_{i=1}^k 
    \mathbb P (N_1=n_1, \dots, N_{i}=n_{i}-1, N_{i+1}=n_{i+1},\ldots, \\
    & \quad N_k=n_k | B ) \cdot
    \mathbb P (\mathcal R (x_n') = i ) \cdot \mathbb P (B=b) \\
 &= \binom{\B}{n_1,n_2,\dots,n_k} \frac{ \gamma^{\B}(1-\gamma)^{n-1-\B} }{k^{\B}}  
    \left[ n_2 (1-\gamma) + \frac{\gamma}{k} (\B+1) \right].
\end{align*}

\subsection{Proof of Theorem~\ref{thm:P_w,Q_w}}

The next theorem gives the random variables we need to calculate 
privacy bounds for the weaker adversary $A_w$:
\begin{thm} \label{Athm:P_w,Q_w}
Assume w.l.o.g. differing elements $x_n=1, x_n'=2$, and 
adversary $A_w$ as given in Def.~\ref{def:A_w}. 
To find a tight DP bound for $\mathcal M= \mathcal{S} \circ \mathcal{R}_{\gamma,k,n}^{PH}$ we can equivalently analyse 
the random variables $P_w, Q_w$ defined as 
\begin{equation} \label{Aeq:Pw_Qw}
P_w = P_1 + P_2, \quad Q_w = Q_1 + Q_2,   
\end{equation}
where 
\begin{equation*}
	\begin{aligned}
		&P_1 \sim (1-\gamma) \cdot  N_1|B, \quad P_2 \sim  \frac{\gamma}{k} \cdot (B+1), \\
		&Q_1 \sim  (1-\gamma) \cdot N_2|B, \quad Q_2 \sim  \frac{\gamma}{k} \cdot (B+1),
	\end{aligned}
\end{equation*}
\begin{equation*} %\label{eq:k_P_Q_2}
\begin{aligned}
& B \sim \mathrm{Bin} (n-1, \gamma), \\
%\text{  } 
&N_i^B|B \sim \mathrm{Bin}(B, 1/k ), \quad i=1,\dots,k, \\
%\begin{aligned}
&N_1|B = N_1^B|B + \mathrm{Bern}(1-\gamma + \gamma/k) \\
% = \mathrm{Bin} (n-1, \gamma/k) + \mathrm{Bern}(1-\gamma + \gamma/k),   \\
&N_2|B = N_2^B|B  + \mathrm{Bern}(\gamma/k).
\end{aligned}
\end{equation*}
\end{thm}
\begin{proof}
Notice that for $k$-RR, 
seeing the shuffler output is equivalent to seeing the total counts for each class 
resulting from applying the local randomisers to $X$ or $X'$. The adversary $A_w$ can  remove all truthfully reported values by client $j$,  $j \in [n-1]$.  
Denote the observed counts after this removal by $n_i, i=1,\dots,k$, so 
$\sum_{i=1}^k n_i = \B+1$. %and write $\mathcal R$ for the local randomiser. 
We now have 
\begin{equation*}
	\begin{aligned}
		 \mathbb P ( \view{ X } = V ) 
		% &= \sum_{i=1}^k 
		%     \mathbb P (N_1=n_1, \dots, N_{i}=n_{i}-1, N_{i+1}=n_{i+1},\dots, N_k=n_k | b ) \cdot
		%     \mathbb P (\mathcal{R}_{\gamma,k,n}^{PH} (x_n) = i ) \cdot \mathbb P (B = b) \\
		&= \sum_{i=1}^k 
		    \mathbb P (N_1=n_1, \dots, N_{i}=n_{i}-1,\dots, N_k=n_k | b ) \cdot 
  \mathbb P (\mathcal{R}_{\gamma,k,n}^{PH} (x_n) = i ) \cdot \mathbb P (B = b) \\
		&= \binom{\B}{n_1,n_2,\dots,n_k} \frac{ \gamma^{\B}(1-\gamma)^{n-1-\B} }{k^{\B}} 
		    \left[ n_1 (1-\gamma) + \frac{\gamma}{k} (\B+1) \right],
	\end{aligned}
\end{equation*}
where the second equation comes from the fact that the random values in $k$-RR 
follow a Multinomial distribution. 
Noting then that $ \mathbb P (\mathcal R_{\gamma,k,n}^{PH}( x_n')=i) = (1-\gamma+\frac{\gamma}{k})$ when $i=2$ and $\frac{\gamma}{k}$ otherwise, 
repeating essentially the same steps gives
\begin{align*}
\mathbb P ( \view{ X'} = V ) 
%    \mathbb P (N_1=n_1, \dots, N_{i}=n_{i}-1,\dots, N_k=n_k | b ) \cdot \\
%&    \mathbb P (\mathcal{R}_{\gamma,k,n}^{PH} (x_n') = i ) \cdot \mathbb P (B=b) \\
 = \binom{\B}{n_1,n_2,\dots,n_k} \frac{ \gamma^{\B}(1-\gamma)^{n-1-\B} }{k^{\B}}  
    \left[ n_2 (1-\gamma) + \frac{\gamma}{k} (\B+1) \right].
\end{align*}
Looking at ratio of the two final probabilities we have
\begin{equation*}
	\begin{aligned}
		\mathbb P_{ V \sim \view{X}} \left[ \frac{ \mathbb P ( \view{ X} = V ) }
		{ \mathbb P ( \view{ X'} = V ) } 
		\geq e^{\veps}\right] =
		 \mathbb P \left[  \frac{ N_1|B \cdot (1-\gamma) + \frac{\gamma}{k} (B+1) }{ N_2|B \cdot (1-\gamma) + \frac{\gamma}{k} (B+1) } \geq e^{\veps} \right],
	\end{aligned}
\end{equation*}
where we write $N_i|B, i\in\{1,2\}$ for the random variable $N_i$ conditional on $B$.  This shows that for DP bounds, the adversaries' full view is equivalent to only considering the joint distribution of $N_i,B,i\in \{1,2\}$, and we can therefore look at the neighbouring random variables
\begin{equation} \label{Aeq:k_P_Q}
P_w = P_1 + P_2, \quad Q_w = Q_1 + Q_2, 
\end{equation}
where
\begin{equation*} 
\begin{aligned}
&P_1 \sim (1-\gamma) \cdot  N_1|B, \quad P_2 \sim  \frac{\gamma}{k} \cdot (B+1), \\
&Q_1 \sim  (1-\gamma) \cdot N_2|B, \quad Q_2 \sim  \frac{\gamma}{k} \cdot (B+1).
\end{aligned}
\end{equation*}
Writing $n_i^B$ for the count in class $i$ resulting from the noise sent by 
the $n-1$ parties, from $k$-RR definition we also have
%writing $N_i^B$ for the contribution to 
%observed counts coming solely from the blanket variables, 
\begin{equation} \label{Aeq:k_P_Q_1}
B \sim \mathrm{Bin} (n-1, \gamma) \quad 
\text{ and } \quad 
N_i^B|B \sim \mathrm{Bin}(B, 1/k ),
\end{equation}
$i=1,\ldots,k.$ As $ V \sim \view{X}$,  we finally have
\begin{equation} \label{Aeq:k_P_Q_2}
\begin{aligned}
&N_1|B = N_1^B|B + \mathrm{Bern}(1-\gamma + \gamma/k)\\
% = \mathrm{Bin} (n-1, \gamma/k) + \mathrm{Bern}(1-\gamma + \gamma/k),   \\
&N_2|B = N_2^B|B  + \mathrm{Bern}(\gamma/k). % = \mathrm{Bin} (n, \gamma/k).
\end{aligned}
\end{equation}
The distributions~\eqref{Aeq:k_P_Q_1} and~\eqref{Aeq:k_P_Q_2} determine the neighbouring distributions $P_w$ and $Q_w$
given in~\eqref{Aeq:k_P_Q} which completes the proof.
\end{proof}

The proof of the following result which 
allows computing tight $\delta(\veps)$-bounds against the adversary $A_w$ for adaptive compositions,
goes analogously to the proof of Thm.~\ref{Athm:P_w,Q_w}.
\begin{thm} %\label{thm:comp_weak}

Consider $m$ compositions of the $k$-RR shuffler mechanism $\mathcal{M}$ and an adversary $A_w$. Then, the tight $(\veps,\delta)$-bound is given by
$$
 \delta(\veps) = \mathbb{P} \left( \sum\limits_{i=1}^m Z_i \geq \veps   \right),
$$
where $Z_i$'s are independent and for all $1 \leq i \leq m$, 
$$
Z_i \sim \log\left( \frac{N_1}{N_2}  \right), \quad 
N_1 \sim P_w,
\quad N_2 \sim Q_w,
$$
where $P_w$ and $Q_w$ are given in \eqref{eq:Pw_Qw}.
\end{thm}

\end{document}